\documentclass{article}
\usepackage{graphicx} 

\usepackage[utf8]{inputenc}
\usepackage{amsmath, amssymb, mathtools, graphicx, amsthm, dsfont, float, bbm, hyperref, caption, adjustbox, kantlipsum, algorithm, authblk, algcompatible, fullpage, subcaption}

\usepackage{xcolor}
\usepackage{comment}
\usepackage{cleveref}
\usepackage{booktabs}
\usepackage[round]{natbib}
\usepackage{paralist}

\newtheorem{theorem}{Theorem}
\newtheorem{proposition}{Proposition}
\newtheorem{lemma}{Lemma}
\theoremstyle{definition}
\newtheorem{definition}{Definition}
\newtheorem{assumption}{Assumption}

\DeclarePairedDelimiter\floor{\lfloor}{\rfloor}
\newcommand{\Cov}{\mathrm{Cov}}
\newcommand{\bbR}{\mathbb{R}}
\newcommand{\bbE}{\mathbb{E}}
\newcommand{\bbP}{\mathbb{P}}
\newcommand{\Var}{\mathrm{Var}}
\newcommand{\st}{\text{ } \vert \text{ } }
\newcommand{\bigst}{\text{ } \Bigg| \text{ } }
\newcommand{\indic}{\mathds{1}}
\newcommand{\indep}{\perp \!\!\! \perp}

\newcommand{\limn}{\underset{n \to \infty}{\mathrm{lim }}}

\newcommand{\strat}{\mathrm{Strat}}
\newcommand{\bern}{\mathrm{Bern}}
\newcommand{\sib}{\mathrm{SIB}}

\newcommand{\tod}{\overset{\mathrm{d}\ }{\to}}

\newcommand{\bsone}{\boldsymbol{1}}

\newcommand{\phm}{\phantom{-}}
\newcommand{\zero}{\ } 
\newcommand{\zcdots}{\ } 
\newcommand{\zvdots}{\ } 

\title{Optimal Correlation for Bernoulli Trials with Covariates}
\author{Tim Morrison and Art B. Owen}
\date{December 2024}

\begin{document}

\maketitle

\begin{abstract}
Given covariates for $n$ units, each of which is to receive a treatment with probability $1/2$, we study the question of how best to correlate their treatment assignments to minimize the variance of the IPW estimator of the average treatment effect. Past work by \cite{bai2022} found that the optimal stratified experiment is a matched-pair design, where the matching depends on oracle knowledge of the distributions of potential outcomes given covariates. We study the strictly broader class of all admissible correlation structures, for which \cite{cytrynbaum2023} recently showed that the optimal design is to divide the units into two clusters and uniformly assign treatment to exactly one of them. This design can be computed by solving a 0-1 knapsack problem that uses the same oracle information. We derive a novel proof of this fact using a result about admissible Bernoulli correlations. We also show how to construct a shift-invariant version of the design that still improves on the optimal stratified design by ensuring that exactly half of the units are treated. A method with just two clusters is not robust to a bad proxy for the oracle, and we mitigate this with a hybrid that uses $O(n^\alpha)$ clusters for $0<\alpha<1$. Under certain assumptions, we also derive a CLT for the IPW estimator under these designs and a consistent estimator of the variance. We compare the proposed designs to the optimal stratified design in simulated examples and find strong performance.  
\end{abstract} 

\section{Introduction}
Consider a setting in which we have access to covariate information $X_i \in \bbR^d$ for $n$ units and run an experiment to estimate the average effect of some binary treatment $Z_i \in \{0, 1\}$. Given treatment assignment, we observe either the response $Y_i(1)$ corresponding to treatment or $Y_i(0)$ corresponding to control. A natural question to ask is whether we can leverage this covariate data to inform our experimental design and ultimately improve treatment effect estimation.

In this paper, we focus specifically on the question of how best to correlate the $Z_i$ terms when each has to be Bernoulli with treatment probability $1/2$. This can be seen as a fairness condition that ensures equal access to treatment. Given access to pre-experiment estimates of the expected responses, \cite{cytrynbaum2023} recently showed how to use a knapsack algorithm to derive a design that would be optimal if we had an oracle version of those estimates. This does not require assuming a parametric form for the data-generating process. We explore this design and further ones inspired by it in greater detail. 

Our work builds on and is directly inspired by \cite{bai2022}, itself a generalization of \cite{barrios2014}. \cite{bai2022} studies how best to stratify experimental units so that exactly $50\%$ of the units in each stratum are selected for treatment. Intuitively, it may be best to pair ``similar" units to guarantee that exactly one receives treatment and one receives control, thereby honing in on the variability caused by treatment. \cite{bai2022} makes this precise by showing that, among all possible stratifications, the optimal one is in fact a matched-pair design in which units are sorted by $g_i = g(X_i) := \bbE[Y_i(1) + Y_i(0) \st X_i]$ and then consecutive units are paired (assuming that $n$ is even), with exactly one unit in each pair receiving treatment.

Bai's (2022) \nocite{bai2022} function $g(\cdot)$ is our oracle.
Although $g(X)$ is never actually known ahead of time, there are several settings in which an experimenter might reasonably believe that they have a good proxy for it.  In A/B testing \citep{koha:tang:xu:2020}, for example, it is common for treatment effects to be very small but nonetheless very valuable. Hence, for a new treatment, a good estimate of $\bbE[Y_i(0) \st X_i]$ from past data provides a good proxy for $g(X_i)/2$. In addition, if $X$ is one-dimensional and there is prior reason to believe that $g(X)$ is monotone in $X$, the optimal stratified design is to pair adjacent units in $X$ regardless of the exact values of $g(X_i)$. Finally, one may have access to observational data relating $Y$ to $X$ and $Z$, which could suffer from confounding bias. Nevertheless, it could still be used to estimate a proxy for $g(X)$, and a matched-pair design run with a reasonable yet inexact estimate may still yield a variance improvement. This provides a way of using observational data to inform the design of an experiment, as in \cite{rosenman2021}.

First, we derive a novel proof of a result shown recently in \cite{cytrynbaum2023} about the optimal design in the strictly broader class of all admissible Bernoulli correlations. In this case, it is optimal to divide the experimental units into disjoint subsets $\mathcal{S}$ and $\mathcal{S}^c$ for which $\sum_{i \in \mathcal{S}} g(X_i)$ and $\sum_{i \in \mathcal{S}^c} g(X_i)$ are as close as possible, and then to assign treatment randomly to exactly one of these two subsets. This is similar to a cluster randomized trial \citep{raudenbush1997, puffer2005}, in which different clusters are independent and units in the same cluster receive the same treatment status; here, though, the two clusters are perfectly anti-correlated, as in a stratified design. Surprisingly, the entire design is determined by a single coin flip.

We also identify and correct two potential flaws in this optimal design.
First, it is not invariant to additive shifts in the potential outcomes by a constant, such as conversion from Celsius to Kelvin. 
To obtain a shift-invariant method, we add the requirement that $|\mathcal{S}| = n/2$ (assuming that $n$ is even). Second, the optimal correlated Bernoulli design is especially sensitive to true knowledge of $g(X)$, so that a bad estimate of $g$ can, in the worst case, cause a large increase in the true variance. Intuitively, this is because the optimal correlated Bernoulli design uses only $2$ random bits whereas the optimal stratified design uses $2^{n/2}$. We propose a hybrid using $2^{\lfloor n^\alpha\rfloor}$ random bits for $0<\alpha<1$, with $\alpha = 1/2$ chosen in simulation.
We show in both theory and simulation that this hybrid still improves on the optimal stratified design while attaining much greater robustness. 

Though our main result also appears there, \cite{cytrynbaum2023} does not discuss these designs in much detail. Rather, that paper primarily studies the unrelated problem of simultaneously performing both sample selection and treatment assignment for an experiment. They employ a two-stage stratification procedure inspired by ideas in \cite{bai2022} and \cite{bai-romano2022}, in which units are first stratified to decide who is enrolled in the experiment and then the enrollees are further stratified to allocate treatment.  

Our work is part of a growing literature on covariate-assisted design of a treatment assignment scheme. One approach to this task is to assume a parametric model for the data such as a linear model \citep{liowenTBD, morrisonowenTBD}, in which case one can write the variance of the treatment effect in terms of an information matrix that depends on the covariates and derive an optimization problem. In the absence of strong modeling assumptions, one may instead try to learn an optimal treatment propensity in an adaptive or sequential setting \citep{hahn2011, offer2021, dai2023,li:owen:2024}. Lastly, much work focuses on ensuring covariate balance between treated and control units (e.g., \cite{rosenberger2008}, \cite{hansen2008}, \cite{li2018}, \cite{morgan2012}, \cite{qin2022}, \cite{ma2020}, \cite{harshaw2024}) and argues that this balance provides some benefit at the estimation stage. 

The optimal correlated Bernoulli design presented here and the optimal stratified design of \cite{bai2022} are similar to covariate balance in that they aim to balance similar units to reduce variance; however, the balance here is based on an oracle estimate (or suitable proxy).
One could instead study a two-stage procedure in which the first stage is meant to learn an optimal stratification in the second stage, as \cite{tabord-meehan2023} does for a class of stratifications that can be represented via decision trees. 

This paper is organized as follows. Section \ref{sec:setup} introduces our notation, estimator, and key assumptions. Section \ref{ssec:opt-corr} outlines our main result on the optimal Bernoulli covariance structure and explains some consequences. Sections \ref{ssec:shift-invar} and \ref{ssec:robust} highlight some imperfections of the optimal correlated Bernoulli design and use them to motivate principled alternatives. Section \ref{sec:asymptotics} presents some results about the asymptotic behavior of our estimator under the optimal correlated Bernoulli design. Finally, Section \ref{sec:sims} presents a simulated example comparing the performance of each design. Proofs can be found in Appendix~\ref{sec:proofs}. Additional simulations are in Appendix~\ref{sec:add_sims}. 

\section{Setup} \label{sec:setup} 
We assume access to $n$ experimental units. For unit $i$, we let $X_i \in \mathcal{X} \subseteq \bbR^d$ denote their covariates, which may be continuous, discrete, or a mixture of the two. We write $Z_i \in \{0, 1\}$ for the $i$'th unit's treatment status, with $Z_i = 1$ corresponding to treatment and $Z_i = 0$ to control. We use $Z\in\{0,1\}^n$ to represent all $n$ treatment values and $X\in\bbR^{n\times d}$ for all of the covariates.

We adopt the potential outcomes framework and make the stable unit treatment value assumption (SUTVA) that $Y_i(Z) = Y_i(Z_i)$. We let $Y_i(1)$ and $Y_i(0)$ denote, respectively, the potential outcomes corresponding to treatment and control for unit $i$. We assume that the tuples $\{(Y_i(1), Y_i(0), X_i)\}$ are drawn IID from some distribution. We let $\mu_1(X_i) = \bbE[Y_i(1) \st X_i]$, $\sigma_1^2(X_i) = \Var(Y_i(1) \st X_i)$, and analogously for $\mu_0(X_i)$ and $\sigma_0^2(X_i)$. 

We focus on the setting in which each $Z_i$ is Bern$(1/2)$ conditional on the covariates (i.e., $\bbP(Z_i = 1 \st X_i) = 1/2$ for all $i \in \{1, 2, \ldots, n\})$ and is uncorrelated with any potential outcomes conditional on the covariates:
$$Z_i \indep \{Y_i(1), Y_i(0)\} \st X_i.$$
This is guaranteed to occur in an experimental design setting in which the experimenter chooses the distribution of $Z_i$ given access only to the covariates. 

However, we allow for the treatment statuses of different units to be correlated with each other. This provides some degree of fairness by ensuring that each unit has the same probability of treatment while still allowing the experimenter to leverage covariate information to improve treatment effect estimation. 

We summarize the class of designs of interest to us via the following definition. 

\begin{definition}\label{def:multivariate-Bern}
A \textit{correlated Bernoulli design} with covariance matrix $\Sigma \in \bbR^{n \times n}$ is a distribution $P$ on $\{0, 1\}^n$ such that, for $Z \sim P$,
\begin{align*} 
\bbP(Z_i = 1) &= \frac{1}{2} \quad\text{and}\quad
\Cov(Z_i, Z_{i'}) = \Sigma_{ii'}. 
\end{align*}
\end{definition}
Of note, we use the term ``correlated Bernoulli design" for the specific case in which each marginal distribution has treatment probability $p = 1/2$, since we do not consider other settings in this paper. We could also assume in Definition \ref{def:multivariate-Bern} that each entry of $\Sigma$ is in $[-1/4, 1/4]$ since that must be the case for a valid covariance matrix for Bern($1/2$) random variables. 

\cite{bai2022} instead studies stratified designs in which the $n$ units are divided into strata and a simple random sample of units in each stratum are given treatment. We formalize this in the following definition. 
\begin{definition}\label{def:stratified}
A \textit{stratified design} is any treatment assignment scheme defined as follows: \\ 
(1) Construct a partition of $\{1, 2, \ldots, n\}$ into disjoint subsets $\{\lambda_1, \lambda_2, \ldots, \lambda_S\}$, with each $|\lambda_s|$ even, which may depend on $\{X_1, X_2, \ldots, X_n\}$. \\ 
(2) Independently for $1 \leq s \leq S$, uniformly at random select $|\lambda_s|/2$ units to receive treatment, and the rest to receive control.
\end{definition}
Note that Definition \ref{def:stratified} assumes even strata sizes and equal numbers of treated and control units in each stratum, which ensures marginal Bern$(1/2)$ distributions for each unit. 

As in \cite{bai2022}, we aim to estimate the average treatment effect conditional on covariates:
$$\tau_n = \frac{1}{n} \sum_{i = 1}^{n} \bbE[Y_i(1) - Y_i(0) \st X_i] := \frac{1}{n} \sum_{i = 1}^{n} \tau(X_i).$$
While our primary interest is on $\tau_n$, in Section \ref{sec:asymptotics} we also report some asymptotic results for $\tau = \bbE[\tau(X_i)].$ Because $\tau=\bbE[\tau_n]$, $\hat\tau_n$ may also be viewed as an estimate of $\tau$.

A natural estimator of $\tau_n$ in this setting is the inverse propensity weighting (IPW) or Horvitz-Thompson estimator \citep{horvitz-thompson}:
\begin{equation} \label{eq:tau-hat} \hat{\tau}_n = \frac{1}{n} \sum_{i = 1}^{n} \frac{Y_iZ_i}{\bbP(Z_i = 1)} - \frac{Y_i(1 - Z_i)}{\bbP(Z_i = 0)} = \frac{2}{n} \sum_{i = 1}^{n} Y_i Z_i - Y_i(1 - Z_i). \end{equation}
Equation \eqref{eq:tau-hat} uses the fact that each unit has a treatment probability of $1/2$ in a correlated Bernoulli design. \cite{bai2022} instead considers the difference-of-means estimator for $\tau_n$. Because there are always $n/2$ treated units for his stratified designs, he arrives at the same final formula for $\hat{\tau}_n$:
$$\hat{\tau}_n^{DM} = \frac{1}{n_1} \sum_{i : Z_i = 1} Y_i - \frac{1}{n_0} \sum_{i : Z_i = 0} Y_i = \frac{2}{n} \sum_{i = 1}^{n} Y_i Z_i - Y_i(1 - Z_i),$$
Here, $n_1 = n/2$ is the number of treated units and $n_0 = n/2$ is the number of control units. Since $n_1$ and $n_0$ need not be fixed in our setting (and, at least a priori, could be zero) it is more natural for us to arrive at $\hat{\tau}_n$ as an IPW estimator.

Note that $\hat{\tau}_n$ in \eqref{eq:tau-hat} is guaranteed to be unbiased conditional on $X$, which follows because each $Z_i$ is marginally Bern$(1/2)$:
\begin{align*} \mathbb{E}[\hat{\tau}_n \st X] &= \frac{2}{n} \sum_{i = 1}^{n} \mathbb{E}[Y_i Z_i - Y_i(1 - Z_i) \st X] \\ 
&= \frac{2}{n} \sum_{i = 1}^{n} \mathbb{E}[Y_i(1) Z_i - Y_i(0)(1 - Z_i) \st X_i] \\
&= \frac{1}{n} \sum_{i = 1}^{n} \mathbb{E}[Y_i(1) - Y_i(0) \st X_i] \\ 
&= \tau_n. \end{align*}
Here, we use in the second equality that $Y_iZ_i = Y_i(1)Z_i$ and $Y_i(1 - Z_i) = Y_i(0)(1 - Z_i)$, and in the third equality that $Z_i$ is Bern$(1/2$) and independent of the potential outcomes given the covariates. 

As in \cite{bai2022}, we aim to minimize the variance of $\hat{\tau}_n$ conditional on the observed covariates: 
$$\Var(\hat{\tau}_n \st X) = \mathbb{E}[(\hat{\tau}_n - \tau_n)^2 \st X].$$
Due to unbiasedness, this is the same as minimizing the conditional mean-squared error. 

In the next subsection, we show how to choose $\Sigma = \Cov(Z)$ to minimize this quantity across all correlated Bernoulli designs. \cite{bai2022} studies the analogous question for selecting strata in a stratified design in the sense of Definition \ref{def:stratified}. Our setting is strictly more general, in that any stratified design corresponds to some choice of $\Sigma$, but there exist valid choices of $\Sigma$ that do not correspond to a stratified design. 
For example,
$$
\Sigma  
=\frac14\begin{bmatrix}
\phm1 & -\frac12 & -\frac12 & \phm0\\[.6ex]
-\frac12 & \phm1 & \phm0 & -\frac12\\[.6ex]
-\frac12 &\phm0 & \phm1 & -\frac12 \\[.6ex]
\phm0 &-\frac12 & -\frac12 & \phm1 \\[.6ex]
\end{bmatrix}
$$
cannot be obtained by stratification.  It comes from an equal probability choice between two stratifications, one with the index pairs $\{(1, 2), (3, 4)\}$ and one with the index pairs $\{(1, 3), (2, 4)\}$.
Another interesting covariance matrix is
$$\Sigma = \frac{1}{4} \begin{bmatrix} \textbf{1}_m\textbf{1}_m^{\top} & -\textbf{1}_m\textbf{1}_{n-m}^{\top} \\ -\textbf{1}_{n-m}\textbf{1}_m^{\top} & \textbf{1}_{n-m}\textbf{1}_{n-m}^{\top} \end{bmatrix},$$
for $1\le m<n$, where $\textbf{1}_d$ is the vector of all ones in $d$ dimensions. This covariance corresponds to the following design: draw a single $Z \sim \text{Bern}(1/2)$, and assign a treatment status of $Z$ to the first $m$ units and a treatment status of $1 - Z$ to the remaining $n - m$ units. This is not a stratified design in the sense of Definition 2.
We show in the next subsection that a particular design of this form is in fact variance-minimizing among all correlated Bernoulli designs.

\section{Design construction}
\subsection{Optimal Bernoulli covariances}\label{ssec:opt-corr}
We begin this subsection by restating (and, for convenience, reproving) a result in \cite{bai2022} that provides a more concrete formula for how $\Var(\hat{\tau}_n \st X)$ depends on a given design. 

\begin{lemma}\label{lemma:equiv-objective}
Let $Z$ be a distribution on $\bbR^n$ whose marginal distributions are each \textnormal{Bern($1/2$)}. For $\Sigma = \Cov(Z)$, 
\begin{align} \label{eq:var-decomp}
\begin{split}
\textnormal{Var}(\hat{\tau}_n \st X) &= \bbE[\textnormal{Var}(\hat{\tau}_n \st X, Z) \st X] + \textnormal{Var}(\bbE[\hat{\tau}_n \st X, Z] \st X) \\ 
&= \frac{2}{n^2} \sum_{i = 1}^{n} \bigl(\sigma_1^2(X_i) + \sigma_0^2(X_i) \bigr)+ \frac{4}{n^2} g^{\top} \Sigma g, 
\end{split}
\end{align}
where $g \in \bbR^n$ is defined by 
\begin{align}\label{eq:defineg}
g_i = \bbE[Y_i(1) + Y_i(0) \st X_i].
\end{align}
\end{lemma}
The proofs of this and most subsequent results are in Appendix \ref{sec:proofs}. Lemma \ref{lemma:equiv-objective} implies that the only property of a correlated Bernoulli design that affects the resulting variance of $\hat{\tau}_n$ is its covariance matrix $\Sigma$, and specifically the quadratic form $g^{\top} \Sigma g$. 

Lemma \ref{lemma:equiv-objective} is stated in terms of the vector $g$, which \cite{bai2022} refers to as the \textit{index function}. Of course, $g$ is never known in practice. However, proxies for $g$, such as those described in the introduction, may bring improvements over a fully randomized design.
Even a poorly-chosen proxy still results in an unbiased estimator. We prove a central limit theorem in Section~\ref{sec:asymptotics} that holds when the proxy is sufficiently accurate.
We return to the role of $g$ and potential consequences of its misestimation in Section \ref{ssec:robust}. 

We are now ready to state our main result, which characterizes the correlated Bernoulli design(s) that minimize $\Var(\hat{\tau}_n \st X)$. 

\begin{theorem}\label{thm:opt-design}
Let $\mathcal{S} \subset \{1, 2, \ldots n\}$ be any subset of indices that solves the following optimization problem:
\begin{align} \label{eq:knapsack}
\begin{split} 
    \underset{S \subset \{1, 2, \ldots, n\}}{\max} \text{ } &\sum_{i \in S} g_i \\ 
    \textnormal{such that } &\sum_{i \in S} g_i \leq \frac{1}{2} \sum_{i = 1}^{n} g_i.
\end{split} 
\end{align}
Consider the design defined as follows: draw a single $Z^* \sim \bern(1/2)$, and assign a treatment status of $Z^*$ to every unit in $\mathcal{S}$ and a treatment status of $1 - Z^*$ to every unit in $\mathcal{S}^c$. This design minimizes $\textnormal{Var}(\hat{\tau}_n \st X)$ among all multivariate designs whose marginal distributions are \textnormal{Bern($1/2$)}. Moreover, 
\begin{equation} \label{final-variance} \mathrm{Var}(\hat{\tau}_n \st X) = \frac{2}{n^2} \sum_{i = 1}^{n}\bigl( \sigma_1^2(X_i) + \sigma_0^2(X_i)\bigr) + \frac{1}{n^2} \left(\sum_{i \in \mathcal{S}} g_i - \sum_{i \in \mathcal{S}^c} g_i \right)^2. \end{equation} 
\end{theorem}
Theorem \ref{thm:opt-design} states that an optimal design divides the units into two clusters for which the sum of $g_i$ is as close as possible, and then assigns treatment with equal probability to exactly one of the two. These subsets need not be of equal size, though in the next subsection we discuss why that is a sensible constraint to add. If multiple subsets solve \eqref{eq:knapsack}, any of them is optimal. 

The optimization problem \eqref{eq:knapsack} is a particular case of the 0-1 knapsack problem in which both the weight vector and value vector are given by $g$ and the capacity is given by half its sum. While NP-complete, this problem has several polynomial-time approximations. See \cite{kellerer} for an overview of the knapsack problem and popular algorithms to solve it.

Theorem \ref{thm:opt-design} is equivalent to Theorem 4.5 in \cite{cytrynbaum2023}. Their construction frames the optimal partition via a Max-Cut problem but is mathematically equivalent to our knapsack formulation. A key difference of our proof is that it uses a result from \cite{huber-maric} that characterizes the set of all admissible correlation structures for correlated Bernoulli distributions. 

For an optimal subset $\mathcal{S}$, let $m = |\mathcal{S}|$. With indices permuted so that units in $\mathcal{S}$ appear before units in $\mathcal{S}^c$, an optimal covariance matrix is of the form 
\begin{equation} \label{eq:Sigma_Bern} \Sigma_{\bern} = \frac{1}{4} \begin{bmatrix} \textbf{1}_{m}\textbf{1}_m^{\top} & -\textbf{1}_m\textbf{1}_{n-m}^{\top} \\ -\textbf{1}_{n-m}\textbf{1}_m^{\top} & \textbf{1}_{n-m}\textbf{1}_{n-m}^{\top} \end{bmatrix}. \end{equation}

As noted in Section \ref{sec:setup}, the class of covariance matrices for stratified designs is a strict subset of the class of covariance matrices for correlated Bernoulli designs, so the resulting variance in Theorem \ref{thm:opt-design} is no worse than that of \cite{bai2022}. Their solution, which assumes that $n$ is even, permutes the units so that $g_{1} \leq g_{2} \leq \ldots \leq g_{n}$ and then stratifies them into the matched pairs $\{(1, 2), (3, 4),\ldots, (n - 1, n)\}$. The resulting design randomly selects one of the two units in each pair to receive treatment and the other to receive control. After permutation, their covariance matrix is therefore block-diagonal with 
$$\Sigma_{\text{strat}} = \frac{1}{4} \begin{bmatrix}
  \phm1 & -1 & \zero & \zero & \zcdots & \zero & \zero \\
 -1 &  \phm1 & \zero & \zero & \zcdots & \zero & \zero \\
  \zero &  \zero & \phm1 & -1 & \zcdots & \zero & \zero \\
  \zero &  \zero & -1 & \phm1 & \zcdots & \zero & \zero \\
  \zvdots & \zvdots & \zvdots & \zvdots & \ddots & \zvdots & \zvdots \\
  \zero &  \zero & \zero & \zero & \zcdots & \phm1 & -1 \\
  \zero &  \zero & \zero & \zero & \zcdots & -1 & \phm1
\end{bmatrix},$$
and 
$$g^{\top} \Sigma_{\text{strat}} g = \frac{1}{4} \sum_{i = 1}^{n/2} (g_{(2i)} - g_{(2i - 1)})^2.$$

By Lemma~\ref{lemma:equiv-objective},
the variance improvement over stratified sampling can be as high as
$$\frac4{n^2}g^{\top} (\Sigma_{\strat} - \Sigma_{\bern}) g=\frac1{n^2}\sum_{i=1}^{n/2}(g_{(2i)}-g_{(2i-1)})^2 - \left(\sum_{i \in \mathcal{S}} g_i - \sum_{i \in \mathcal{S}^c} g_i \right)^2.
$$
For given values of $g_{(1)}$ and $g_{(n)}$, this ranges from $0$ to $(g_{(n)} - g_{(1)})^2/n^2$. The latter case occurs when both the optimal knapsack difference is zero and all of but one of the increments $g_{(2i)} - g_{(2i - 1)}$ are zero. For instance, this happens when $g = [1, 1, \ldots, 1, n - 1]^{\top} \in \bbR^n$ and $\mathcal{S}=\{1,2,\dots,n-1\}$.

We can gain further intuition about the formula for the optimal correlated Bernoulli design by considering the design-dependent second term of the variance decomposition in \eqref{eq:var-decomp}. 
\cite{bai2022} calls that term 
the variance of the \textit{ex-post bias}, since we can write it as $\Var\bigl(\bbE[\hat{\tau}_n-\tau_n \st X, Z] \st X\bigr)$.
Though $\hat{\tau}_n$ is unbiased on average across all $Z$, this need not be the case conditional on any particular draw of $Z$. However, the optimal correlated Bernoulli design minimizes this variance by grouping units so that $\bbE[\hat{\tau}_n \st X, Z]$ is as tightly concentrated as possible. 

In the special case in which $g$ has the same sum in $\mathcal{S}$ and $\mathcal{S}^c$, so that $g^{\top} \Sigma_{\bern}g = 0$, the following proposition states that the ex-post bias is in fact zero. 

\begin{proposition}\label{prop:noexpostbias}
For the optimal correlated Bernoulli design,
suppose that  $\sum_{i\in\mathcal{S}}g_i=\sum_{i\in\mathcal{S}^c}g_i$ holds for $g$ defined in ~\eqref{eq:defineg}.
Then with $z \in \{0, 1\}^n$ such that $z_i = \indic\{i \in \mathcal{S}\}$, 
\begin{align}\label{eq:flipz}
\bbE[\hat\tau_n\st X,Z=z]=\bbE[\hat\tau_n\st X,Z=\bsone_n-z].
\end{align}
Hence, $\Var(\bbE[\hat\tau_n\st X,Z]\st X)=0$, and $\bbE[\hat{\tau}_n \st X, Z] = \tau_n$ with probability one in $Z$. 
\end{proposition}
\begin{proof}
Define $z\in\{0,1\}^n$ by $z_i = \indic\{i\in\mathcal{S}\}$. Then
\begin{align*} 
\bbE[\hat{\tau}_n \st X, Z=z] &= \frac{2}{n} \sum_{i = 1}^{n} \bbE[Y_i(1)z_i - Y_i(0)(1 - z_i) \st X, Z=z] \\ 
&= \frac{2}{n} \sum_{i : z_i = 1} \bbE[Y_i(1) \st X_i] - \frac{2}{n} \sum_{i : z_i = 0} \bbE[Y_i(0) \st X_i] \\ 
&= \frac{2}{n} \sum_{i : z_i = 1} \bigl( \bbE[Y_i(1) \st X_i] + \bbE[Y_i(0) \st X_i] - \bbE[Y_i(0) \st X_i]\bigr)  \\ 
&\quad- \frac{2}{n} \sum_{i : z_i = 0} \bigl( \bbE[Y_i(1) \st X_i] + \bbE[Y_i(0) \st X_i] - \bbE[Y_i(1) \st X_i]\bigr) \\ 
&= \frac{2}{n} \sum_{i : z_i = 0} \bbE[Y_i(1) \st X_i] - \frac{2}{n} \sum_{i : z_i = 1} \bbE[Y_i(0) \st X_i] \\ 
&= \bbE[\hat{\tau}_n \st X, Z = \bsone_n - z]
\end{align*} 
establishing~\eqref{eq:flipz}. In the penultimate equality, we used that the sums of $g$ in $\mathcal{S}$ and $\mathcal{S}^c$ are the same.
The optimal correlated Bernoulli design
has $\Pr(Z=z \st X)=\Pr(Z=\bsone_n-z \st X)=1/2$.
Because the expectation in~\eqref{eq:flipz} is the same for either treatment assignment, we get $\Var(\bbE[\hat\tau_n\st X,Z]\st X)=0$.
Moreover, because $\hat{\tau}_n$ is unbiased for $\tau_n$ conditional on $X$, each term in \eqref{eq:flipz} must equal $\tau_n$, and so the ex-post bias is zero in this special case. 
\end{proof}

\subsection{Shift invariance} \label{ssec:shift-invar}
Recall that the optimal correlated Bernoulli design chooses a subset $\mathcal{S} \subseteq \{1, 2, \ldots, n\}$ of indices that minimizes the squared knapsack difference 
\begin{equation} \label{eq:knapsack-diff} \left(\sum_{i \in \mathcal{S}} g_i - \sum_{i \in \mathcal{S}^c} g_i \right)^2, \end{equation}
where $g_i = \bbE[Y_i(1) + Y_i(0) \st X_i]$. Conveniently, this solution is invariant to scaling of all potential outcomes by the same constant, such as converting the response variable from centimeters to inches. However, it is not necessarily invariant to additive shifts of the potential outcomes by the same constant, such as converting from Celsius to Kelvin, since
$$\left(\sum_{i \in \mathcal{S}} (g_i + c) - \sum_{i \in \mathcal{S}^c} (g_i + c)\right)^2 = \left(\sum_{i \in \mathcal{S}} g_i - \sum_{i \in \mathcal{S}^c} g_i + (|\mathcal{S}| - |\mathcal{S}^c|) c\right)^2.$$
This may result in a new solution unless we add the constraint that $|\mathcal{S}| = |\mathcal{S}^c| = n/2$ (with $n$ even). Fundamentally, this occurs because the Horvitz-Thompson/IPW estimator is not shift-invariant either, even though the true $\tau_n$ is. 

Since it is often sensible that the design should be invariant to linear shifts of the response variable, we now consider the optimal shift-invariant Bernoulli design, defined by:
\begin{align} \label{eq:balanced-knapsack}
\begin{split} 
    \underset{S \subset \{1, 2, \ldots, n\}}{\max} \text{ } &\sum_{i \in S} g_i \\ 
    \textnormal{such that } &\sum_{i \in S} g_i \leq \frac{1}{2} \sum_{i = 1}^{n} g_i, \\ 
    &|\mathcal{S}| = \frac{n}{2}.
\end{split} 
\end{align}
This problem is often known as set partitioning \citep{karmarkar1983} or balanced (two-way) number partitioning. Like the knapsack problem, it is NP-complete but has several heuristic algorithms \citep{zhang2011}. We refer to the resulting optimal covariance matrix as $\Sigma_{\sib}$ (for ``shift-invariant Bernoulli"). 

Desirably, a solution to \eqref{eq:balanced-knapsack} still improves on the optimal stratified design. We state this as the following theorem. 
\begin{theorem}\label{thm:invariant-wins}
For $n$ even, let $\mathcal{S}$ be a subset that solves \eqref{eq:balanced-knapsack}. Then 
$$\left(\sum_{i \in \mathcal{S}} g_i - \sum_{i \in \mathcal{S}^c} g_i \right)^2 \leq \sum_{i = 1}^{n/2} (g_{(2i)} - g_{(2i - 1)})^2.$$
In other words, the variance from the optimal shift-invariant Bernoulli design is no larger than the variance from the optimal stratified design. 
\end{theorem}
The proof of Theorem \ref{thm:invariant-wins} involves constructing an explicit balanced partition of $\{1, 2, \ldots, n\}$ for which the desired inequality holds; the same must then be true of the optimal balanced partition. This partition is defined sequentially as follows: for each pair $(g_{(2k - 1)}, g_{(2k)})$ with $k \leq n$, add one of them to $\mathcal{S}$ and the other to $\mathcal{S}^c$. To do so, simply add $g_{(2k)}$ to whichever of the two currently has a smaller total sum. As shown in the proof in Appendix \ref{sec:proofs}, this gives a lower quadratic form than the optimal stratified design. 

\subsection{Dependence on oracle knowledge} \label{ssec:robust}
We now explore the robustness of $\Sigma_{\bern}$ and $\Sigma_{\sib}$ to incorrect choices of $g$, which in practice is unknown as an oracle quantity. Intuitively, an optimal correlated Bernoulli design should be more sensitive to the choice of $g$ than the optimal stratified design, since the former is determined by a single coin flip, whereas the latter involves $n/2$ coin flips. 

Of note, \cite{cytrynbaum2023} similarly cautions against using a plug-in estimate of $g$ from a small pilot study. They argue that this should be strictly worse than simply doing an IID Bern$(1/2)$ treatment assignment and then incorporating the covariates at the estimation stage via a cross-fit AIPW estimator. 

Our first result shows that, if $g$ is incorrectly estimated by some amount, then the variance of $\hat{\tau}_n$ from using $\Sigma_{\bern}$ can be larger than expected by an amount quadratic in $n$, while this inflation can only be linear in $n$ from using $\Sigma_{\text{strat}}$. 
This result is formulated in terms of a worst-case perturbation $h$ of the true $g$ by up to some error $\epsilon$ in each coordinate. 

\begin{theorem}\label{thm:l^infty-result}
Suppose $\Sigma_{\bern}$ and $\Sigma_{\strat}$ are the covariance matrices of the optimal correlated Bernoulli and optimal stratified designs, respectively, for index function $h \in \bbR^n$. Let $B_{\epsilon}^{\infty}(h)$ be the $\ell^{\infty}$ ball around $h$ of radius $\epsilon$. Then
\begin{align*} 
&\underset{g \in B_{\epsilon}^{\infty}(h)}{\sup} \text{ } g^{\top} \Sigma_{\bern} g \geq h^{\top} \Sigma_{\bern} h + n^2 \epsilon^2,\quad\text{and}  \\ 
&\underset{g \in B_{\epsilon}^{\infty}(h)}{\sup} \text{ } g^{\top} \Sigma_{\strat} g = h^{\top} \Sigma_{\strat} h + 2n\epsilon^2 + 4\epsilon(h_{(n)} - h_{(1)}).
\end{align*}
In other words, 
even if all $h_i$ are within $\epsilon$ of $g_i$,
$g^{\top} \Sigma_{\bern}g$ can be larger than expected by a term of order $O(n^2\epsilon^2)$ while $g^{\top} \Sigma_{\strat}g$ is only off by a term of order $O(n\epsilon^2)$. 
\end{theorem}

Theorem \ref{thm:l^infty-result} shows that we should be wary of using $\Sigma_{\bern}$ from some inaccurate $h$, since the limited randomness can lead to an overconfident design. The extra term $h_{(n)} - h_{(1)}$ is likely to scale at a slower than linear rate; for instance, it is $O(1)$ if the distribution of $h$ is bounded and $o_p(\sqrt{n})$ if the distribution of $h$ has finite variance (see Lemma \ref{lemma:max-o_p} in Appendix \ref{sec:proofs}). The result is the exact same for $\Sigma_{\sib}$, since the proof of Theorem \ref{thm:l^infty-result} does not depend on $|\mathcal{S}|$, so enforcing $|\mathcal{S}| = n/2$ does not fix this problem. 

A natural follow-up question is whether there is a correlated Bernoulli design that is more robust to this misestimation. Specifically, we would like a design that improves on the optimal stratified design for the true $g$ while getting closer to the robustness of the optimal stratified design in the sense of Theorem \ref{thm:l^infty-result}. 

We consider the following class of hybrid designs as an interpolation. First, permute indices so that $g$ is sorted from smallest to largest. Next, partition $\{1, 2, \ldots, n\}$ into $G = \floor{n^{\alpha}}$ contiguous groups of indices for some $\alpha \in (0, 1)$. Each group will have roughly $n^{1 - \alpha}$ indices, but add or remove one index from groups as needed to make all group cardinalities even. Then, run the optimal correlated Bernoulli design independently in each group. The resulting covariance matrix $\Sigma_{\mathrm{hybrid}}$ will be block diagonal with $G$ blocks, each of rank one and (up to permutation within the group) of the form \eqref{eq:Sigma_Bern}. This requires solving $G$ separate knapsack problems, but each is of much lower time complexity. To adapt this to the shift-invariant design, one need only solve the balanced partitioning problem \eqref{eq:balanced-knapsack} instead of the knapsack problem for each group. This procedure is summarized in Algorithm \ref{alg:hybrid-design}. 

\begin{algorithm}[t] 
\caption{Construction of hybrid designs}\label{alg:hybrid-design}
\begin{algorithmic}[1]
\STATE{ \textbf{Input:} even sample size $n > 0$, interpolation parameter $\alpha \in (0, 1)$, oracle vector $g \in \bbR^n$} 
\STATE{ \textbf{Sort} $g$ so that $g_1 \leq g_2 \leq \ldots \leq g_n$}
\STATE{ \textbf{Compute} $G = \floor{n^{\alpha}}$, $k = 2 \floor{n/(2G)}$, and $r = (n - kG)/2$}
\STATE{ \textbf{Partition} $\{1, 2, \ldots, n\}$ into $G$ contiguous subgroups $\{\mathcal{G}_1, \ldots \mathcal{G}_G\}$ with 
$$|\mathcal{G}_{\ell}| = \begin{cases} k + 2, & \ell \leq r \\ 
                                  k, & \ell > r 
                    \end{cases}$$
}
\FOR{$\ell$ in $\{1, 2, \ldots, G$\}}
\STATE{ \textbf{Define} $g_{\mathcal{G}_{\ell}} \in \bbR^{|\mathcal{G}_{\ell}|}$ to be the subvector of $g$ with indices in $\mathcal{G}_{\ell}$} 
\STATE{ \textbf{Solve} the knapsack problem \eqref{eq:knapsack} or the balanced number partitioning problem \eqref{eq:balanced-knapsack} to obtain $\mathcal{S}_{\mathcal{G}_{\ell}}$, $\mathcal{S}^c_{\mathcal{G}_{\ell}}$}
\STATE{ \textbf{Sample} $Z_{\ell} \sim \bern(1/2)$}
\STATE{ \textbf{Define} $Z_{\mathcal{G}_{\ell}} \in \{0, 1\}^{|\mathcal{G}_{\ell}|}$ by 
$$Z_{G_{\ell}, i} = \begin{cases} 
                    Z_{\ell}, & i \in \mathcal{S}_{\mathcal{G}_{\ell}} \\ 
                    1 - Z_{\ell}, & i \in \mathcal{S}_{\mathcal{G}_{\ell}}^c
              \end{cases}$$
}
\ENDFOR
\STATE{\textbf{Output:} $Z = (Z_{\mathcal{G}_1}, Z_{\mathcal{G}_2}, \ldots, Z_{\mathcal{G}_G}) \in \{0, 1\}^n$}
\end{algorithmic}    

\smallskip

N.B.: at step 4, we can use any $G$ contiguous subgroups
of which $r$ have cardinality $k+2$ and $G-r$ have 
cardinality $k$.
\end{algorithm}

For both the optimal correlated Bernoulli design and its shift-invariant version, any hybrid design will still yield a lower variance than the optimal stratified design when using the true $g$. To see this, note that $g^{\top} \Sigma_{\mathrm{hybrid}} g$ is additive over the blocks of $\Sigma_{\mathrm{hybrid}}$. By the proof of Theorem \ref{thm:invariant-wins}, the quadratic form restricted to each block is no larger than the corresponding quadratic form of the optimal stratified design restricted to that block. Here we use that each group has even cardinality and the groups are contiguous partitions of the sorted $g$, so we can be sure that no matched pair in the optimal stratified design is split across two separate groups. 

To assess the worst-case variance inflation under misestimation of $g$ (in the $\ell^{\infty}$ sense of Theorem \ref{thm:l^infty-result}), it suffices to consider the inflation separately in each group since $g^{\top} \Sigma_{\mathrm{hybrid}} g$ is additive across groups. Following the steps in the proof of Theorem \ref{thm:l^infty-result}, the worst-case variance inflation will be of order $n^{2 (1 - \alpha)} \epsilon^2$ in each group. The worst-case overall variance inflation will then be of order $n^{2(1 - \alpha)} n^{\alpha} \epsilon^2 = n^{2 - \alpha} \epsilon^2$. Hence, the robustness increases with the number of groups, approaching the quadratic inflation of the optimal correlated Bernoulli design as $\alpha \to 0$ and the linear inflation of the optimal stratified design as $\alpha \to 1$. These designs thus naturally interpolate between the two extremes. 

\section{Asymptotic behavior}\label{sec:asymptotics}
In this section, we present some results about the asymptotic behavior of $\hat{\tau}_n$ under the optimal correlated Bernoulli design. Here, the target of estimation shifts from $\tau_n$ to $\tau = \bbE[Y_i(1) - Y_i(0)]$, the population average treatment effect that is not conditional on the observed $X$.

We focus on the shift-invariant case in which $n/2$ units receive treatment, so we assume that $n$ is even and interpret limits to be along this subsequence. We state our results in terms of a given choice of known index function $h \in \bbR^n$, which need not be the true $g$. However, the results require some degree of accuracy in capturing the true $g$, which we make more precise below. Throughout, we let $\mathcal{S}_h$ and $\mathcal{S}_h^c$ be the solutions to the balanced partitioning problem \eqref{eq:balanced-knapsack} obtained using $h$. 

Before proceeding, we state several assumptions about the distribution of $\{Y_i(1), Y_i(0), X_i\}$ and the subset $\mathcal{S}_h$ that we will need. 

\begin{assumption} \label{asymptotic-assumptions} 
\text{} \\ 
(a) $\bbE[Y_i^2(z)] < \infty$ for $z \in \{0, 1\}$. \\ 
(b) $0 < \bbE[\Var(Y_i(z) \st X_i)]$ for $z \in \{0, 1\}$. \\ 
(c) For $z \in \{0, 1\}$, 
\begin{equation} \label{eq:assumption3} 
    \frac{1}{n} \Biggl|\sum_{i \in \mathcal{S}_h} \Var(Y_i(z) \st X_i) - \sum_{i \in \mathcal{S}_h^c} \Var(Y_i(z) \st X_i) \Biggr| \overset{p}{\to} 0. 
\end{equation}
(d) With $g(X_i) = \bbE[Y_i(1) + Y_i(0) \st X_i]$, 
\begin{equation} \label{eq:assumption4} 
    \frac{1}{\sqrt{n}} \Biggl|\sum_{i \in \mathcal{S}_h} g(X_i) - \sum_{i \in \mathcal{S}_h^c} g(X_i) \Biggr| \overset{p}{\to} 0. 
\end{equation}
\end{assumption}
The first condition is a standard regularity condition on the distribution of the potential outcomes, and the second ensures that there is still some variability left in the potential outcomes after conditioning on the covariates. The third and fourth conditions are more unusual and state that the first and second (conditional) moments of the potential outcomes take on similar averages in the subsets $\mathcal{S}_h$ and $\mathcal{S}_h^c$ asymptotically. The condition in \eqref{eq:assumption3} holds, for instance, if $\sigma^2_1(X_i)$ and $\sigma^2_0(X_i)$ are homoscedastic. In addition, it is guaranteed to hold under any covariate-agnostic choice of subsets such as $\mathcal{S}_h = \{1, 3, 5, \ldots, n - 1\}$ and $\mathcal{S}_h^c = \{2, 4, 6, \ldots, n\}$ by the law of large numbers. 

To assess the reasonableness of \eqref{eq:assumption4}, we note that it holds when $h = g$ and the distribution of $g(X_i) = \bbE[Y_i(1) + Y_i(0) \st X_i]$ has finite variance. To see this, first assume that the distribution of $g(X_i)$ is bounded, so $|g(X_i)| \leq M$ almost surely. Consider the following sequential construction of a balanced partition: for each new pair $g_{2k - 1}$ and $g_{2k}$, add the larger one to whichever of $\mathcal{S}$ and $\mathcal{S}^c$ has a smaller sum, and vice-versa. This will keep the gap between the sums in $[-2M, 2M]$. In the unbounded case, the same construction gives a gap that is in $[-2g_{(n)}, 2g_{(n)}]$. Since the maximum of $n$ IID draws from a finite variance distribution is $o_p(\sqrt{n})$ (see Lemma \ref{lemma:max-o_p} in Appendix \ref{sec:proofs}), \ref{eq:assumption4} then holds.  

The various conditions in Assumption \ref{asymptotic-assumptions} are very similar to those used in \cite{bai-romano2022, bai-jiang2024, bai-liu2024}, which study inference for matched-pair and matched-tuple designs. In place of \eqref{eq:assumption3} and \eqref{eq:assumption4}, they assume that $\bbE[Y_i^r(z) \st X_i]$ is Lipschitz for $z \in \{0, 1\}$ and $r \in \{1, 2\}$ and that the average within-pair $\ell^1$ and $\ell^2$ covariate distances converge to zero as $n \to \infty$. 

Under our stated assumptions, the next theorem establishes the asymptotic normality of $\hat{\tau}_n$. We write $\mu_1 = \bbE[Y_i(1)]$, $\mu_0 = \bbE[Y_i(0)]$, $\sigma_1^2 = \Var(Y_i(1))$, and $\sigma_0^2 = \Var(Y_i(0))$ for the population-level moments of the potential outcomes. 

\begin{theorem} \label{thm:CLT}
Suppose that $Z$ is sampled from an optimal correlated shift-invariant Bernoulli design and all conditions in Assumption \ref{asymptotic-assumptions} are satisfied. Then 
$$\sqrt{n}(\hat{\tau}_n - \tau) \tod N(0, \nu^2),$$
where 
\begin{align} \label{eq:nu^2-definition}
\begin{split}
\nu^2 &= 2\bbE[\sigma_1^2(X_i) + \sigma_0^2(X_i)] + \bbE[(\tau(X_i) - \tau)^2] \\ 
&= 2\sigma_1^2 + 2\sigma_0^2 - \bbE[(g(X_i) - \bbE[g(X_i)])^2].
\end{split}
\end{align}
\end{theorem} 
The proof of this result is very similar to that of Lemma S.1.4 in \cite{bai-romano2022}. The formula for the asymptotic variance is in fact the exact same, though we have an extra factor of two in our definition of $\nu^2$ because we divide by $n$ instead of $2n$; we have $n$ observations where they have $n$ pairs of observations. Hence, both designs yield the same asymptotic performance, so the improvements over the optimal stratified design are purely a finite sample phenomenon. 

A natural follow-up question is how to consistently estimate $\nu^2$ in Theorem \ref{thm:CLT} to construct a confidence interval for $\tau$. To do so, we will first need a slightly different version of conditions (c) and (d) of Assumption \ref{asymptotic-assumptions}. 

\begin{assumption} \label{asymptotic-assumptions2} 
For $r \in \{1, 2\}$ and $z \in \{0, 1\}$, 
\begin{equation*} 
    \frac{1}{n} \Biggl|\sum_{i \in \mathcal{S}_h} \bbE[Y_i^r(z) \st X_i] - \sum_{i \in \mathcal{S}_h^c} \bbE[Y_i^r(z) \st X_i] \Biggr| \overset{p}{\to} 0. 
\end{equation*}
\end{assumption}

In Lemma \ref{lemma:moment-consistency} in Appendix \ref{sec:proofs}, we show that we can then consistently estimate the moments $\bbE[Y_i^r(z)]$ for $r \in \{1, 2\}$ and $z \in \{0, 1\}$ via their sample analogs. This implies that we can estimate the means $\mu_1$ and $\mu_0$ and variances $\sigma_1^2$ and $\sigma_0^2$ consistently. As a result, we could always overestimate $\nu^2$ with $2\hat{\sigma}_1^2 + 2\hat{\sigma}_0^2$ and get a confidence interval that is asymptotically conservative. 

However, the terms in $\nu^2$ that involve variances of conditional expectations or expectations of conditional variances are trickier to estimate. To do so, we need to be able to isolate the behavior at or near a particular $X_i$ more effectively. In light of these considerations, we again adapt an idea from \cite{bai-romano2022}, who are able to estimate $\nu^2$ by considering sample quantities defined using matched-pairs, such as the average squared difference in within-pair outcomes. Since exactly one unit in each pair is treated, and since pairs have similar covariates asymptotically by assumption, this provides additional localized information for estimating $\nu^2$. 

Our design is not based on matched pairs, but we now state an assumption that will allow us to leverage these ideas. To do so, we let $\pi$ be the permutation of $\{1, 2, \ldots, n\}$ so that $\{\pi(1), \pi(2), \ldots, \pi(n/2)\}$ are the indices in $\mathcal{S}_h$ and 
\begin{align*} &h(X_{\pi(1)}) \leq h(X_{\pi(2)}) \leq \ldots \leq h(X_{\pi(n/2)}), \\ 
&h(X_{\pi(n/2 + 1)}) \leq h(X_{\pi(n/2 + 2)}) \leq \ldots \leq h(X_{\pi(n)}).\end{align*} 
That is, $\pi$ sorts the units so that the first $n/2$ are in $\mathcal{S}_h$, the last $n/2$ are in $\mathcal{S}_h^c$, and the units within each half are sorted by $h(X_i)$. 

\begin{assumption} \label{pairing-assumption}
With $\pi$ defined as above, 
\begin{equation} 
\frac{1}{n} \sum_{i = 1}^{n/2} \bigl(\mu_z(X_{\pi(i)}) - \mu_z(X_{\pi(i + n/2)})\bigr)^2 \overset{p}{\to} 0
\end{equation}
for $z \in \{0, 1\}$. 
\end{assumption}

\begin{assumption} \label{pairs-of-pairs-assumption}
With $\pi$ defined as above, 
\begin{equation} 
\frac{1}{n} \sum_{i = 1}^{n/2} \bigl(\mu_z(X_{\pi(2i)}) - \mu_z(X_{\pi(2i -1 )}\bigr)^2
\end{equation}
for $z \in \{0, 1\}$. 
\end{assumption}

Assumptions \ref{pairing-assumption} and \ref{pairs-of-pairs-assumption} create an analog of matched pairs in our setting that will allow us to estimate $\nu^2$ consistently. They state respectively that the conditional mean functions become sufficiently close (1) within a pair and (2) for units in the same partition within adjacent pairs. 

With these assumptions, we can derive a consistent estimator $\hat{\nu}^2$ for $\nu^2$, which we state as the following theorem.

\begin{theorem} \label{thm:consistency}
Suppose that Assumptions \ref{asymptotic-assumptions}, \ref{asymptotic-assumptions2}, \ref{pairing-assumption}, and \ref{pairs-of-pairs-assumption} hold. Let
\begin{align*} \label{eq:nu^2-definition}
\begin{split}
\hat{a}_n^2 &= \frac{2}{n} \sum_{i = 1}^{n/2} (Y_{\pi(i)} - Y_{\pi(i + n/2)})^2, \\ 
\hat{b}_n^2 &= \frac{4}{n} \sum_{i = 1}^{n/2} Y_{\pi(2i)} Y_{\pi(2i - 1)} - \frac{4}{n} \sum_{i = 1}^{n/2} Y_{\pi(i)} Y_{\pi(i + n/2)}. \\
\end{split}
\end{align*} 
Further, let $\hat{\nu}^2 = 2\hat{a}_n^2 - (\hat{b}_n^2 + \hat{\tau}_n^2)$. For $\nu^2$ given in~\eqref{eq:nu^2-definition}, 
$$\hat{\nu}^2 \overset{p}{\to} \nu^2\quad\text{and}\quad\frac{\sqrt{n}(\hat{\tau}_n - \tau)}{\hat{\nu}} \tod N(0, 1).$$
\end{theorem}
The estimator $\hat{\nu}^2$ utilizes information about outcomes within the same pair and adjacent pairs. This is similar to the adjusted t-test estimator in \cite{bai-romano2022}, with the formula for $\hat{b}_n^2$ adapted to fit our correlated Bernoulli design. The proof of Theorem \ref{thm:consistency} essentially involves showing that our assumptions are sufficient for the same argument used in their Theorem 3.3. We note that the idea to use adjacent pairs to estimate components of the variance also appears in \cite{abadie2008}. 

\section{Simulation results}\label{sec:sims}

\begin{figure}[t]
\centering
\includegraphics[width=0.8\textwidth]{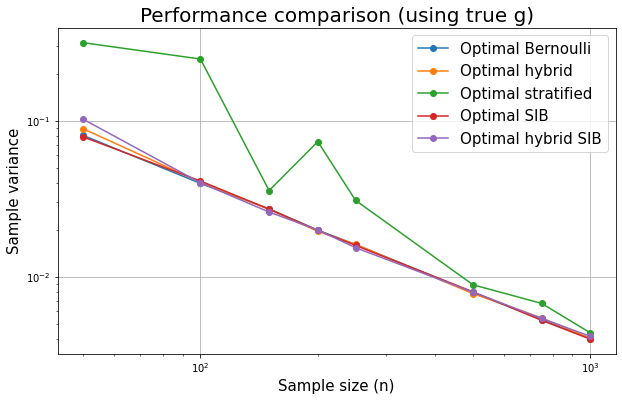}
\caption{Comparison of designs when computed with the true $g$. Designs shown are the optimal correlated Bernoulli design, optimal hybrid Bernoulli design with $\floor{\sqrt{n}}$ groups, optimal shift-invariant Bernoulli (SIB) design, optimal hybrid SIB design with $\floor{\sqrt{n}}$ groups, and optimal stratified design. Both sample size and sample variance are plotted on the log scale for ease of visibility, with sample sizes of $n \in \{50, 100, 150, 200, 250, 500, 750, 1000\}$.}
\label{fig:noiseless}
\end{figure}

In this section, we present a simulation example to compare the performances of the various designs discussed thus far. Our simulation setup is as follows. For given sample size $n$, we generate $X$ and the conditional mean functions $\mu_1(X)$ and $\mu_0(X)$ via 
\begin{align*} 
&X \sim N\left(\begin{bmatrix} 0 \\ 0 \end{bmatrix}, \begin{bmatrix} 10 & 5 \\ 5 & 10 \end{bmatrix} \right), \\
&\mu_1(X) = X_1^2 - 3 |X_2|^{3/2}, \\ 
&\mu_0(X) = -2 |X_2|^{3/2}. 
\end{align*}
We then take $g(X) = \mu_1(X) + \mu_0(X)$. After generating $Z$ from a given design, we add Gaussian noise to the conditional mean function to generate $Y$, i.e., 
$$Y_i = \mu_{Z_i}(X_i) + N(0, \sigma_Y^2).$$
Note that this implies $\sigma_1^2(X) = \sigma_0^2(X) = \sigma_Y^2$ is constant in both $X$ and $Z$; since the optimal design does not depend on these variances at all, as seen in Lemma \ref{lemma:equiv-objective}, we opt for homoscedastic noise for simplicity. Throughout, we take $\sigma_Y^2 = 1$, giving a fundamental variance lower bound of $1/n$. 

We begin by comparing performance of the designs when $g(X)$ is known exactly. To do so, we sample $X$ as above once for a given $n$, then run each design $10{,}000$ times and compute the sample variance of the resulting $\hat{\tau}_n$ estimator for each design. Since the target variance is conditional on $X$, each simulation uses the same $X$ (and thus $g$) but still has randomness in $Z$ and $Y$. 

We use the Python package PuLP \citep{PuLP} with the mixed-integer programming solver Gurobi \citep{Gurobi} to solve the necessary knapsack and balanced partitioning problems. Of course, this is not guaranteed to produce the optimal solution, but we find it both quite fast and quite performant in terms of producing a small knapsack difference. 

\begin{figure}[t]
    \begin{subfigure}[b]{0.32\textwidth}
        \centering
        \includegraphics[width=\textwidth]{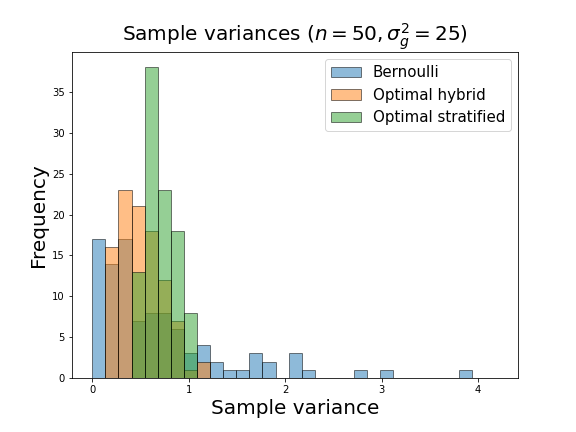}
    \end{subfigure}
    \hfill
    \begin{subfigure}[b]{0.32\textwidth}
        \centering
        \includegraphics[width=\textwidth]{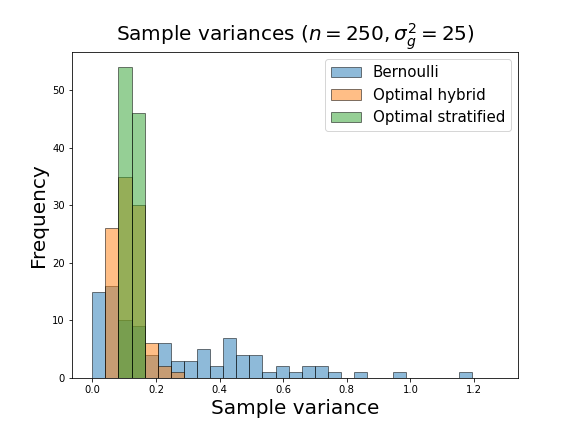}
    \end{subfigure}
    \hfill
    \begin{subfigure}[b]{0.32\textwidth}
        \centering
        \includegraphics[width=\textwidth]{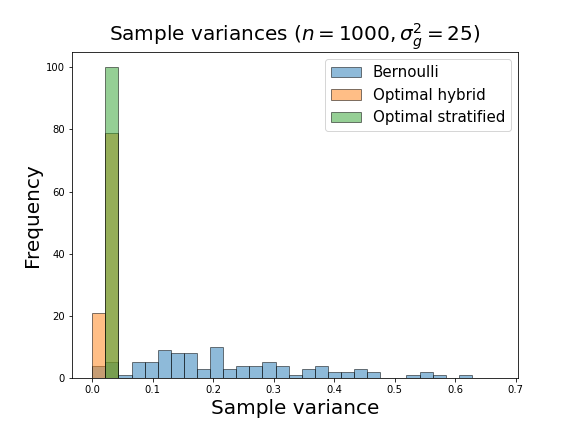}
    \end{subfigure} 

    \begin{subfigure}[b]{0.32\textwidth}
        \centering
        \includegraphics[width=\textwidth]{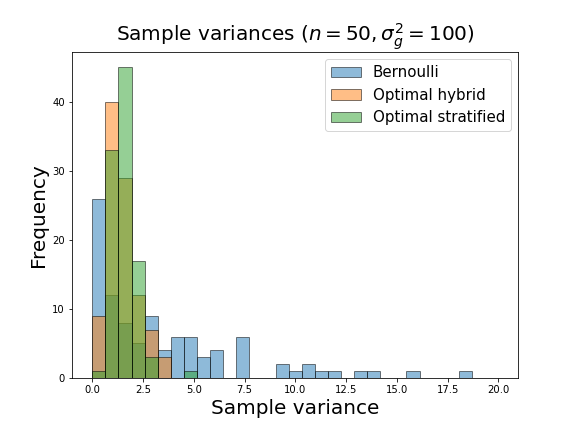}
    \end{subfigure}
    \hfill
    \begin{subfigure}[b]{0.32\textwidth}
        \centering
        \includegraphics[width=\textwidth]{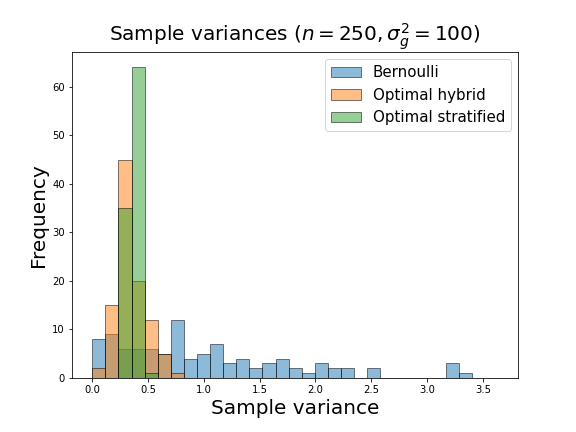}
    \end{subfigure}
    \hfill
    \begin{subfigure}[b]{0.32\textwidth}
        \centering
        \includegraphics[width=\textwidth]{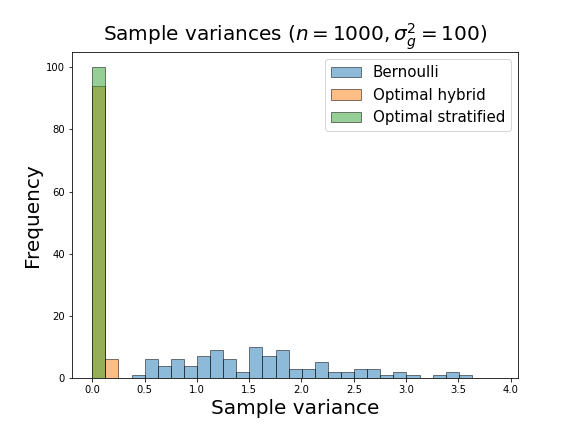}
    \end{subfigure} 

    \begin{subfigure}[b]{0.32\textwidth}
        \centering
        \includegraphics[width=\textwidth]{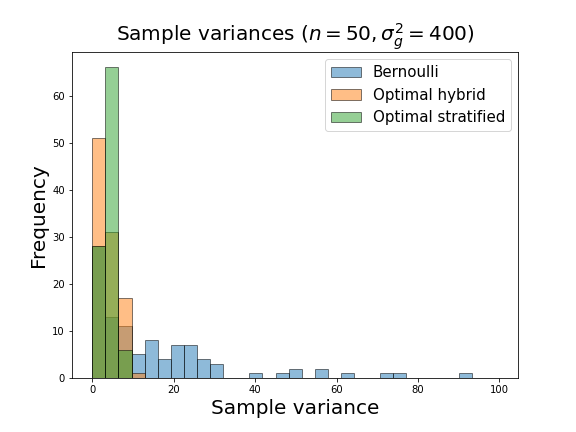}
    \end{subfigure}
    \hfill
    \begin{subfigure}[b]{0.32\textwidth}
        \centering
        \includegraphics[width=\textwidth]{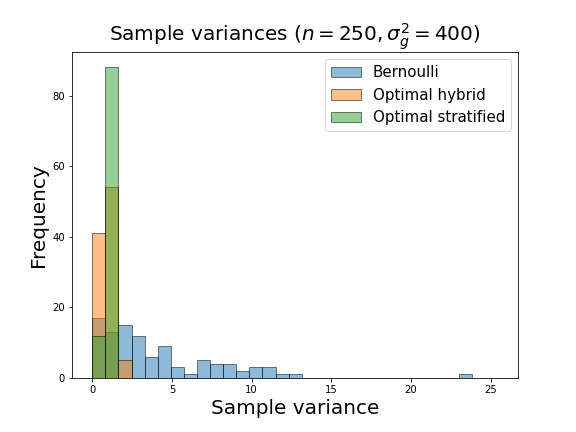}
    \end{subfigure}
    \hfill
    \begin{subfigure}[b]{0.32\textwidth}
        \centering
        \includegraphics[width=\textwidth]{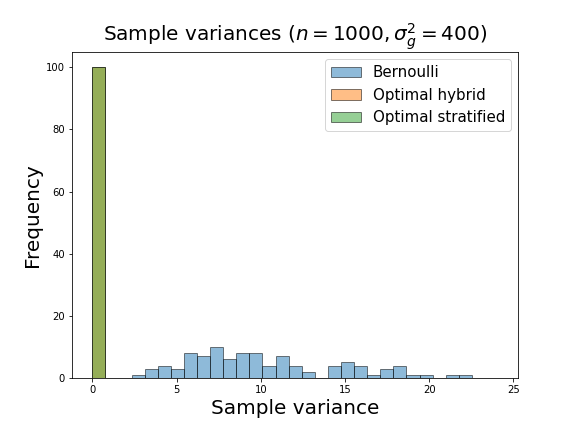}
    \end{subfigure} 
    
\caption{Distribution of sample variances across $100$ perturbations $h = g + N(0, \sigma_g^2)$. The three rows are at variances of $\sigma_g^2 \in \{25, 100, 400\}$, and the three columns are at samples sizes of $n \in \{50, 250$, $1000$\}.}
\label{fig:perturb}
\end{figure}

\begin{figure}
\centering
\includegraphics[width=0.6\textwidth]{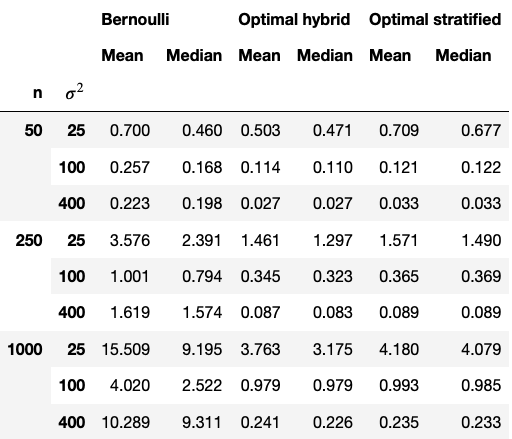}
\caption{Mean and median variances across the $100$ perturbations in Figure \ref{fig:perturb}, grouped by design and simulation parameters.}
\label{fig:mean_median}
\end{figure}

Figure \ref{fig:noiseless} plots the sample variance as a function of $n$ for each of: the optimal correlated Bernoulli design, the optimal shift-invariant Bernoulli design, the hybrid analogues described in Section \ref{ssec:robust}, and the optimal stratified design of \cite{bai2022}. As expected, the Bernoulli designs outperform the stratified design when $g$ is known exactly. In addition, the shift-invariant design and all hybrid designs perform far more comparably to the optimal correlated Bernoulli design than to the optimal stratified design. We note also that this performance improvement shrinks as $n$ increases. This is to be expected since, for all designs, the term dominating the variance as $n$ increases becomes the design-independent term from Lemma \ref{lemma:equiv-objective}. 

Next, we assess robustness of these same designs to incorrect knowledge of the oracle information. Specifically, we consider the same simulation setup but where our estimate $h$ is a perturbation of the true $g$. To move beyond the worst-case analysis of Theorem \ref{thm:l^infty-result}, we study the distribution of the sample variances over many such perturbations, defined via $h = g + N(0, \sigma_g^2).$

For a given perturbation, we run each design $1000$ times, as in the previous analysis, and compute the sample variance for that $h$. For fixed $\sigma_g^2$, we do this across $100$ of these perturbations, obtaining a distribution of $100$ sample variances for each design. Figure \ref{fig:perturb} shows the results of this procedure across $9 = 3\times 3$ settings: $n \in \{50, 250, 1000\}$ and $\sigma_g^2 \in \{25, 100, 400\}$ (for reference, the sample variance of the vector $g$ itself is typically between $500$ and $700$ for this simulation setup). To avoid cluttering the plots, we remove the shift-invariant and hybrid shift-invariant designs, which behave analogously to the optimal correlated Bernoulli designs. 

Figure \ref{fig:perturb} plots the results from these simulations, and Figure \ref{fig:mean_median} shows the mean and median variance for each design across all simulation settings. Both confirm the behavior that we would expect from the analysis in Section \ref{ssec:robust}. The optimal correlated Bernoulli design is very sensitive to mistaken choice of $g$, so increases in $\sigma_g^2$ lead to significant increases in the variance under that design. Moreover, as we saw in the worst-case behavior in Theorem \ref{thm:l^infty-result}, the variance of the optimal correlated Bernoulli design under an incorrect $h$ worsens as $n$ increases. In contrast, the optimal hybrid design shows far more stability to these perturbations and tends to maintain variances near those of the optimal stratified design even for large $n$ and large $\sigma_g^2$. 

These experiments provide further evidence that the hybrid designs can retain many of the desired properties of both the optimal correlated Bernoulli design and optimal stratified design. Appendix \ref{sec:add_sims} provides several other simulation settings to investigate when the performance of our designs is particularly strong or weak, as well as a performance comparison when using the proxy $h = \mu_0$. 

\section*{Acknowledgments} 
The authors thank John Cherian, Kevin Guo, Harrison Li, Anav Sood, and James Yang for helpful discussions and suggestions. T.\ M.\ was partly supported by a B.\ C.\ and E.\ J.\ Eaves Stanford Graduate Fellowship. This work was supported by the National Science Foundation under grant DMS-2152780.

\newpage 
\bibliographystyle{plainnat}
\bibliography{bernoulli}

\appendix
\section{Proofs}\label{sec:proofs}

\begin{proof}[\bf Proof of \Cref{lemma:equiv-objective}]
By the law of total variance, 
\begin{equation} \label{eq:total-variance}
\Var(\hat{\tau}_n \st X) = \bbE[\Var(\hat{\tau}_n \st X, Z) \st X] + \Var(\bbE[\hat{\tau}_n \st X, Z] \st X).
\end{equation}
We expand the first term:
\begin{align*} 
\bbE[\Var(\hat{\tau}_n \st X, Z) \st X] &= \bbE\left[\Var\left(\frac{2}{n} \sum_{i = 1}^{n} Y_i(1)Z_i - Y_i(0)(1 - Z_i) \st X, Z \right) \bigst  X \right] \\ 
&= \frac{4}{n^2} \bbE\left[\sum_{i = 1}^{n} \Var(Y_i(1)Z_i - Y_i(0)(1 - Z_i) \st X, Z) \bigst X \right] \\ 
&= \frac{4}{n^2} \bbE\left[\sum_{i = 1}^{n} Z_i \Var(Y_i(1) \st X, Z) + (1 - Z_i)\Var(Y_i(0) \st X, Z) \bigst X \right] \\ 
&= \frac{4}{n^2} \bbE\left[\sum_{i = 1}^{n} Z_i \Var(Y_i(1) \st X_i) + (1 - Z_i) \Var(Y_i(0) \st X_i) \bigst X \right] \\ 
&= \frac{2}{n^2} \sum_{i = 1}^{n} \sigma_1^2(X_i) + \sigma_0^2(X_i). 
\end{align*}
The second equality uses that the potential outcomes are independent for different units conditional on $X$ and $Z$. The third equality expands out the variance and uses that $Z_i(1 - Z_i) = 0$ identically to remove the cross-term. The fourth equality changes the conditioning to just $X_i$ since by assumption that it is the only term in $(X, Z)$ that affects the potential outcome distributions. Finally, the last equality uses that $\bbE[Z_i \st X] = 1/2$. The resulting quantity is the fundamental lower bound on the variance that arises due to the variability of the potential outcomes around their means. 

We now turn to the second term in \eqref{eq:total-variance}. We can similarly expand to obtain:
\begin{align*} 
\Var(\bbE[\hat{\tau}_n \st X, Z] \st X) &= \Var\left(\bbE\left[\frac{2}{n} \sum_{i = 1}^{n} Y_i(1)Z_i - Y_i(0)(1 - Z_i) \st X, Z \right] \bigm| X \right) \\ 
&= \frac{4}{n^2} \Var\left(\sum_{i = 1}^{n} \bbE[(Y_i(1) + Y_i(0))Z_i - Y_i(0) \st X, Z ] \bigm| X \right) \\ 
&= \frac{4}{n^2} \Var\left(\sum_{i = 1}^{n} Z_i \bbE[Y_i(1) + Y_i(0) \st X, Z] - \bbE[Y_i(0) \st X, Z ] \bigm| X \right) \\
&= \frac{4}{n^2} \Var\left(\sum_{i = 1}^{n} Z_i \bbE[Y_i(1) + Y_i(0) \st X, Z ] \bigm| X \right) \\ 
&= \frac{4}{n^2} \Var(g^{\top} Z) \\ 
&= \frac{4}{n^2} g^{\top} \Sigma g, 
\end{align*} 
where we recall that $\Sigma = \Cov(Z)$. The third equality uses that $Z_i$ is constant when conditioning on it, and the fourth uses that $Y_i(0)$ is independent of $Z_i$ given $X_i$, so $\bbE[Y_i(0) \st X_i, Z_i]=\bbE[Y_i(0)\st X_i]$, which has zero variance conditional on $X_i$. 
\end{proof}

\begin{proof}[\bf Proof of \Cref{thm:opt-design}]
We first state an important lemma about correlated Bernoulli distributions that is the main theorem of \cite{huber-maric}. To do so, fix a vector $v \in \{0, 1\}^n$ and let $\textbf{1}_n - v \in \{0, 1\}^n$ be the vector with opposite entries to those of $v$. Let $P_v$ be the two-point distribution that is uniform on $\{v, \textbf{1}_n - v\}$. Then $P_v$ induces marginal Bern($1/2$) distributions in each coordinate. The ensuing lemma states that any valid correlation structure on correlated Bernoulli distributions arises as a convex combination of these two-point distributions. 

\begin{lemma}\label{lemma:huber-maric}[Theorem 1 of \cite{huber-maric}] Let $\Sigma$ be the covariance matrix of a multivariate Bernoulli
distribution $P$ whose marginals are all Bern$(1/2)$. Then there exists $P'$ that is the convex combination of distributions of the form \textnormal{Unif($\{v, \textbf{1}_n - v\}$)} such that the covariance matrix of $P'$ is $\Sigma$. 
\end{lemma}

Lemma \ref{lemma:huber-maric} allows us to write an arbitrary covariance matrix for $n$ Bern($1/2$) random variables as a convex combination of all $2^n$ covariance matrices for distributions of the form $P_v$ with $v \in \{0, 1\}^n$.  
Our convex combination will place weight $\delta_v$ on $P_v$, and we use $\boldsymbol{\delta}$ to denote all $2^n$ weights $\delta_v$.
In light of Lemma \ref{lemma:equiv-objective}, we thus aim to solve
\begin{align} \label{eq:conv-comb}
\min_{\boldsymbol{\delta}} \text{ } &g^{\top} \Biggl(\,\sum_{v\in \{0,1\}^n} \delta_v \Cov(P_{v}) \Biggr) g \\
\,\text{such that } &\text{$\delta_v\ge0$ for all $v\in\{0,1\}^n$}, \\ 
&\sum_{v\in\{0,1\}^n}\delta_v=1.\notag
\end{align}
Clearly, \eqref{eq:conv-comb} is solved by placing all the weight on a vector $v$ that minimizes $g^{\top} \Cov(P_v) g$ (of which there could be several). For $v \in \{0, 1\}^n$ and $Z \sim P_v$, we have
\begin{align*} \Cov(Z_i, Z_{i'}) &= \bbP(Z_i = Z_{i'} = 1) - \frac{1}{4} \\
&= \frac{1}{2} v_i v_{i'} + \frac{1}{2} (1 - v_i)(1 - v_{i'}) - \frac{1}{4}\\ 
&=  \begin{cases} 
        \frac{1}{4}, & v_i = v_{i'} \\ 
        -\frac{1}{4}, & v_i \neq v_{i'}. 
    \end{cases}
\end{align*}
In matrix form, this is
$$\Cov(Z) = \frac{1}{2} vv^{\top} + \frac{1}{2} (1 - v)(1 - v)^{\top} - \frac{1}{4} \textbf{1}_n \textbf{1}_n^{\top}.$$
Hence, for $v \in \{0, 1\}^n,$
\begin{align*} 
g^{\top} \Cov(P_v) g &= g^{\top} \left(\frac{1}{2} vv^{\top} + \frac{1}{2} (1 - v)(1 - v)^{\top} - \frac{1}{4} \textbf{1}_n \textbf{1}_n^{\top} \right) g \\ 
&= \frac{1}{2} \left( (g^{\top} v)^2 + (g^{\top}(1 - v))^2 - \frac{1}{2} (g^{\top}\textbf{1}_n)^2 \right) \\ 
&= \frac{1}{2} \left( \left(\sum_{i : v_i = 1} g_i\right)^2 + \left(\sum_{i : v_i = 0} g_i\right)^2 - \frac{1}{2} \left(\sum_{i = 1}^{n} g_i\right)^2\right) \\ 
&= \frac{1}{4} \left( \sum_{i : v_i = 1} g_i - \sum_{i : v_i = 0} g_i \right)^2.
\end{align*} 
From this formula, it is apparent that the best choice of $v$ is any for which $\sum_{i : v_i = 1} g_i$ and $\sum_{i : v_i = 0} g_i$ are as close as possible. That is precisely the optimization problem \eqref{eq:knapsack} in Theorem \ref{thm:opt-design}. Finally, from the proof of Lemma \ref{lemma:equiv-objective}, the optimal total variance is 
\begin{align*} \Var(\hat{\tau}_n \st X) &= \frac{2}{n^2} \sum_{i = 1}^{n} \sigma_1^2(X_i) + \sigma_0^2(X_i) + \frac{4}{n^2} g^{\top} \Sigma_{\bern} g \\ 
&= \frac{2}{n^2} \sum_{i = 1}^{n} \sigma_1^2(X_i) + \sigma_0^2(X_i) + \frac{1}{n^2} \left( \sum_{i \in \mathcal{S}} g_i - \sum_{i \in \mathcal{S}^c} g_i \right)^2, \end{align*} 
where $\mathcal{S}$ is any subset that solves \eqref{eq:knapsack}. 
\end{proof}

\begin{proof}[\bf Proof of \Cref{thm:invariant-wins}]
Begin by defining $\delta_k = g_{(k)} - g_{(k - 1)}$ for $1 \leq k \leq n$, so that each $\delta_k \geq 0$. We can then write $g_{(k)} = \delta_1 + \ldots + \delta_k$. To prove the desired result, we show inductively that, for any even $k \leq n$, there exists a partition of the indices $\{(1),\ldots, (k)\}$ into disjoint subsets $\mathcal{S}_k$ and $\mathcal{S}_k^c$ of equal cardinality so that 
\begin{equation} \label{eq:inductive-ineq} \left(\sum_{i \in \mathcal{S}_k} g_{(i)} - \sum_{i \in \mathcal{S}_k^c} g_{(i)} \right)^2 \leq \sum_{i = 1}^{k/2} (g_{(2i)} - g_{(2i - 1)})^2. \end{equation}
In this way, we greedily construct a sequence of partitions with the desired property up to any even index $k$, and taking $k = n$ would give the desired result. 

For $k = 2$, the result is immediate since there is only one partition and so \eqref{eq:inductive-ineq} is an equality. For $k = 4$, consider the partition $\mathcal{S}_k = \{(1), (4)\}$ and $\mathcal{S}_k^c = \{(2), (3)\}$. Then 
\begin{align*} 
\left(\sum_{i \in \mathcal{S}_k} g_i - \sum_{i \in \mathcal{S}_k^c} g_i \right)^2 &= (\delta_4 - \delta_2)^2 
\leq \delta_2^2 + \delta_4^2 
= \sum_{i = 1}^{k/2} (g_{(2i)} - g_{(2i - 1)})^2.
\end{align*} 
Now suppose that we have a balanced partition satisfying \eqref{eq:inductive-ineq} for even $k < n$ and wish to construct one for $k + 2$. Let 
$$D_k = \sum_{i \in \mathcal{S}_k} g_i - \sum_{i \in \mathcal{S}_k^c} g_i.$$
We will add $g_{(k + 1)}$ to one sum and $g_{(k + 2)}$ to the other, so the net result will be either to add or subtract $\delta_{k + 2}$ from the total difference $D_k$. Consider appending index $(k + 2)$ to $\mathcal{S}_k$ if $D_k$ is negative and appending index $(k + 1)$ to $\mathcal{S}_k$ if $D_k$ is non-negative. Let $\text{sgn}(D_k) \in \{-1, 1\}$ be the sign of $D_k$. Then this procedure yields
\begin{align*} D_{k + 2}^2 &= (D_k - \text{sgn}(D_k)\delta_{k + 2})^2 \\ 
&= D_k^2 + \delta_{k + 2}^2 - 2D_k\text{sgn}(D_k) \delta_{k + 2}\\ 
&\leq D_k^2 + \delta_{k + 2}^2 \\ 
&\leq \delta_2^2 + \delta_4^2 + \ldots + \delta_k^2 + \delta_{k + 2}^2 \\ 
&= \sum_{i = 1}^{(k + 2)/2} (g_{(2i)} - g_{(2i - 1)})^2. 
\end{align*} 
Here, we use that the term being subtracted in the second line is always non-negative, and we use the inductive hypothesis in the last inequality. Taking $k = n$, the partition we have sequentially constructed retains an objective that is no larger than that of the optimal stratified design. Since this is merely one balanced partition, the same must be true for the solution to \eqref{eq:balanced-knapsack}. 

\end{proof}

\begin{proof}[\bf Proof of \Cref{thm:l^infty-result}]
Let $\mathcal{S}$ and $\mathcal{S}^c$ solve the knapsack problem \eqref{eq:knapsack} for the index function $h$. Then 
$$h^{\top} \Sigma_{\bern} h = \left( \sum_{i \in \mathcal{S}^c} h_i - \sum_{i \in \mathcal{S}} h_i \right)^2,$$
and the former sum is no smaller than the latter sum (they may be equal). If we can perturb each $h_i$ by up to $\epsilon$, this quantity is clearly made largest by increasing each $h_i$ in $\mathcal{S}^c$ by $\epsilon$ and decreasing each $h_i$ in $\mathcal{S}$ by $\epsilon$. The result is 
\begin{align*} 
g^{\top} \Sigma_{\bern} g &= \left( \sum_{i \in \mathcal{S}^c} (h_i + \epsilon) - \sum_{i \in \mathcal{S}} (h_i - \epsilon) \right)^2 \\ 
&= \left(\sum_{i \in \mathcal{S}^c} h_i - \sum_{i \in \mathcal{S}} h_i + n\epsilon \right)^2 \\ 
&= h^{\top} \Sigma_{\bern} h + n^2\epsilon^2 + 2n\epsilon \left(\sum_{i \in \mathcal{S}^c} h_i - \sum_{i \in \mathcal{S}} h_i\right) \\ 
&\geq h^{\top} \Sigma_{\bern} h + n^2\epsilon^2. 
\end{align*} 
Meanwhile, 
$$h^{\top} \Sigma_{\text{strat}} h = \sum_{i = 1}^{n/2} (h_{(2i)} - h_{(2i - 1)})^2,$$
and so the perturbation with the largest effect is to increase each term of the form $h_{(2k)}$ by $\epsilon$ and decrease each term of the form $h_{(2k - 1)}$ by $\epsilon$. This gives
\begin{align*} 
g^{\top} \Sigma_{\bern} g &= \sum_{i = 1}^{n/2} (h_{(2i)} - h_{(2i - 1)} + 2\epsilon)^2 \\ 
&= \sum_{i = 1}^{n/2} (h_{(2i)} - h_{(2i - 1)})^2 + 4n\epsilon^2 + 4\epsilon (h_{(n)} - h_{(1)}) \\ 
&= h^{\top} \Sigma_{\text{strat}} h + 4n\epsilon^2 + 4\epsilon (h_{(n)} - h_{(1)}). 
\end{align*}
Therefore, the inflation is of order $n^2\epsilon^2$ for the Bernoulli design and only $n\epsilon^2$ for the stratified design.  
\end{proof} 

\begin{proof}[\bf Proof of \Cref{thm:CLT}]
Our proof is very similar to that of Lemma S.1.4 in \cite{bai-romano2022}. Recall that 
$$\hat{\tau}_n = \frac{2}{n} \sum_{i = 1}^{n} Y_i(1) Z_i - Y_i(0)(1 - Z_i).$$
We decompose $\sqrt{n}(\hat{\tau}_n - \tau)$ into 
\begin{equation} \label{eq:tau_n-decomposition}
\sqrt{n}(\hat{\tau}_n - \tau) = A_n - B_n + C_n - D_n,
\end{equation}
where 
\begin{align*} 
A_n &= \frac{2}{\sqrt{n}} \sum_{i = 1}^{n} Y_i(1) Z_i - \bbE[Y_i(1) Z_i \st X, Z], \\ 
B_n &= \frac{2}{\sqrt{n}} \sum_{i = 1}^{n} Y_i(0) (1 - Z_i) - \bbE[Y_i(0) (1 - Z_i) \st X, Z], \\ 
C_n &= \frac{2}{\sqrt{n}} \sum_{i = 1}^{n} \bbE[Y_i(1) Z_i \st X, Z] - Z_i \bbE[Y_i(1)],\quad\text{and} \\ 
D_n &= \frac{2}{\sqrt{n}} \sum_{i = 1}^{n} \bbE[Y_i(0) (1 - Z_i) \st X, Z] - (1 - Z_i) \bbE[Y_i(0)]. 
\end{align*} 
Here, we use that exactly $n/2$ of the $Z_i$ terms are one and the rest are zero. Conditional on $X$ and $Z$, $A_n$ and $B_n$ are independent of each other and $C_n$ and $D_n$ are constant. We first analyze the limiting behavior of each term separately, then argue that we can combine them. 

We begin with $A_n$. Conditional on $X$ and $Z$, the individual terms in the sum are independent but not identically distributed. We verify that the Lindeberg central limit theorem holds conditional on $X$ and $Z$. Let
\begin{align*} s_n^2 &= \sum_{i = 1}^{n} \Var(2Y_i(1) Z_i - 2 \bbE[Y_i(1) Z_i \st X, Z] \st X, Z) \\ 
&= 4 \sum_{i = 1}^{n} \Var(Y_i(1) Z_i \st X, Z) \\ 
&= 4 \sum_{i: Z_i = 1} \Var(Y_i(1) \st X, Z) \\ 
&= 4 \sum_{i : Z_i = 1} \Var(Y_i(1) \st X_i). 
\end{align*}
In the last line, we use SUTVA and the assumption that the potential outcomes are conditionally independent of treatment given covariates. Now, 
$$\frac{s_n^2}{n} = \frac{2}{n} \sum_{i = 1}^{n} \Var(Y_i(1) \st X_i) + \frac{2}{n} \left( \sum_{i : Z_i = 1} \Var(Y_i(1) \st X_i) - \sum_{i : Z_i = 0} \Var(Y_i(0) \st X_i) \right).$$
By condition (c) in Assumption \ref{asymptotic-assumptions}, the second term in the preceding sum is $o_p(1)$, so 
\begin{align*} \frac{s_n^2}{n} &= \frac{2}{n} \sum_{i = 1}^{n} \Var(Y_i(1) \st X_i) + o_p(1) \overset{p}{\to} 2\bbE[\Var(Y_i(1) \st X_i)] \\
&= 2 \bbE[\sigma_1^2(X_i)]. \end{align*} 
Here, we use the law of large numbers, which holds because of condition (a) in Assumption \ref{asymptotic-assumptions}. For arbitrary $\epsilon > 0$, the Lindeberg condition is that 
\begin{equation} \label{eq:Lindeberg-condition} \frac{1}{s_n^2} \sum_{i = 1}^{n} \bbE\Bigl[\bigl(Y_i(1) Z_i - \bbE[Y_i(1) Z_i \st X, Z]\bigr)^2 \indic\bigl\{|Y_i(1) Z_i - \bbE[Y_i(1) Z_i \st X, Z]| > \epsilon s_n\bigr\} \st X, Z \Bigr] \overset{p}{\to} 0. \end{equation}
Fix any $M < \infty$. Since $\epsilon s_n \overset{p}{\to} \infty$, we have
\begin{align*} 
& \quad \text{ } \frac{1}{s_n^2} \sum_{i = 1}^{n} \bbE\Bigl[\bigl(Y_i(1) Z_i - \bbE[Y_i(1) Z_i \st X, Z]\bigr)^2 \,\indic\bigl\{|Y_i(1) Z_i - \bbE[Y_i(1) Z_i \st X, Z]| > \epsilon s_n\bigr\} \st X, Z \Bigr] \\ 
&\leq \frac{n}{s_n^2} \frac{1}{n} \sum_{i = 1}^{n} \bbE\Bigl[\bigl(Y_i(1) - \bbE[Y_i(1) \st X_i]\bigr)^2\, \indic\bigl\{|Y_i(1) - \bbE[Y_i(1) \st X_i]| > \epsilon s_n\bigr\} \st X, Z \Bigr] \\ 
&= \frac{n}{s_n^2} \frac{1}{n} \sum_{i = 1}^{n} \bbE\Bigl[\bigl(Y_i(1) - \mu_1(X_i)\bigr)^2\, \indic\bigl\{|Y_i(1) - \mu_1(X_i)| > \epsilon s_n\bigr\} \st X, Z \Bigr] \\ 
&\leq \frac{n}{s_n^2} \frac{1}{n} \sum_{i = 1}^{n} \bbE\Bigl[\bigl(Y_i(1) - \mu_1(X_i)\bigr)^2\, \indic\bigl\{|Y_i(1) - \mu_1(X_i)| > M\bigr\} \st X, Z \Bigr] + o_p(1)  \\ 
&= \frac{n}{s_n^2} \frac{1}{n} \sum_{i = 1}^{n} \bbE\Bigl[\bigl(Y_i(1) - \mu_1(X_i)\bigr)^2 \,\indic\bigl\{|Y_i(1) - \mu_1(X_i)| > M\bigr\} \st X_i \Bigr] + o_p(1) \\ 
&\overset{p}{\to} \frac{1}{2\bbE[\sigma_1^2(X_i)]} \bbE\Bigl[\bigl(Y_i(1) - \mu_1(X_i)\bigr)^2\, \indic\bigl\{|Y_i(1) - \mu_1(X_i)| > M\bigr\}\Bigr].
\end{align*} 
The first inequality uses that the summand is zero whenever $Z_i = 0$, the second uses that $\epsilon s_n > M$ for all $n$ sufficiently large, and the final line uses Slutsky's theorem and our earlier formula for the convergence of $s_n^2/n$. Since $(Y_i(1) - \bbE[Y_i(1) \st X_i])^2$ is integrable by condition (a) of Assumption \ref{asymptotic-assumptions}) and upper bounds the integrand in the final line across all $M$, we can apply the dominated convergence theorem as $M \to \infty$ and conclude that \eqref{eq:Lindeberg-condition} is true. The Lindeberg condition therefore holds in probability. 

From this, we claim that 
\begin{equation} 
\label{eq:CDF-convergence}
\underset{t \in \bbR}{\sup} \text{ } \Bigl|\bbP(A_n \leq t \st X, Z) - 
\Phi\Bigl(\frac{t}{\sqrt{2\bbE[\sigma_1^2(X_i)]}}\Bigr)\Bigr| \overset{p}{\to} 0.
\end{equation}
Suppose by way of contradiction that there exists a subsequence along which \eqref{eq:CDF-convergence} fails. Since the expression in \eqref{eq:Lindeberg-condition} converges to zero in probability, there exists a further subsequence along which it converges with probability one in $(X, Z)$. Along this further subsequence, then, $A_n \tod N(0, 2\bbE[\Var(Y_i(1) \st X_i)])$, so \eqref{eq:CDF-convergence} holds for this subsequence. This is a contradiction, so in fact \eqref{eq:CDF-convergence} must hold for the original sequence. 

An analogous proof shows that
$$\underset{t \in \bbR}{\sup} \text{ } \Bigl|\bbP(B_n \leq t \st X, Z) - \Phi\Bigl(\frac{t}{\sqrt{2\bbE[\sigma_0^2(X_i)]}}\Bigr)\Bigr| \overset{p}{\to} 0.$$
Since $A_n$ and $B_n$ are independent conditional on $X$ and $Z$, a similar subsequencing argument shows that 
$$\underset{t \in \bbR}{\sup} \text{ } \Bigl|\bbP(A_n - B_n \leq t \st X, Z) - \Phi\Bigl(\frac{t}{\sqrt{2\bbE[\sigma_1^2(X_i) + \sigma_0^2(X_i)]}}\Bigr)\Bigr| \overset{p}{\to} 0.$$

We now move on to the terms $C_n$ and $D_n$ in the decomposition \eqref{eq:tau_n-decomposition}, which we handle as a pair. We have 
$$C_n - D_n = \frac{2}{\sqrt{n}} \sum_{i = 1}^{n} Z_i\bigl(\mu_1(X_i) - \mu_1\bigr) - (1 - Z_i)\bigl(\mu_0(X_i) - \mu_0\bigr).$$
Then
$$\bbE[C_n - D_n \st X] = \frac{1}{\sqrt{n}} \sum_{i = 1}^{n} (\mu_1(X_i) - \mu_1) - (\mu_0(X_i) - \mu_0).$$
In addition, since the vector $Z$ takes on only two values with equal probability conditional on $X$, 
\begin{align*} 
\Var(C_n - D_n \st X) &= \bbE[(C_n - D_n - \bbE[C_n - D_n \st X])^2 \st X] \\ 
&= \frac{2}{n} \left(\sum_{i \in \mathcal{S}} \mu_1(X_i) - \mu_1 - \sum_{i \in \mathcal{S}^c} \mu_0(X_i) - \mu_0 \right. \\ 
&\left. \quad \quad \quad - \frac{1}{2} \sum_{i = 1}^{n} (\mu_1(X_i) - \mu_1) - (\mu_0(X_i) - \mu_0)\right)^2 \\ 
&\quad + \frac{2}{n} \left(\sum_{i \in \mathcal{S}^c} \mu_1(X_i) - \mu_1 - \sum_{i \in \mathcal{S}} \mu_0(X_i) - \mu_0 \right. \\ 
&\left. \quad \quad \quad - \frac{1}{2} \sum_{i = 1}^{n} (\mu_1(X_i) - \mu_1) - (\mu_0(X_i) - \mu_0)\right)^2 \\ 
&= \frac{2}{n} \left(\sum_{i \in \mathcal{S}} \mu_1(X_i) + \mu_0(X_i) - \sum_{i \in \mathcal{S}^c} \mu_1(X_i) + \mu_0(X_i) \right)^2 \\ 
&= \frac{2}{n} \left(\sum_{i \in \mathcal{S}} g(X_i) - \sum_{i \in \mathcal{S}^c} g(X_i) \right)^2 \\ 
&\overset{p}{\to} 0. 
\end{align*} 
Here, we use that $|\mathcal{S}| = |\mathcal{S}^c|$ to cancel the $\mu_1$ and $\mu_0$ terms in the third equality, and we use condition (d) of Assumption \ref{asymptotic-assumptions} in the last line. Therefore, 
$$\bbP(|C_n - D_n - \bbE[C_n - D_n \st X]| > \epsilon \st X) \overset{p}{\to} 0$$
by Chebyshev's inequality. By the dominated convergence theorem (using that all probabilities are bounded by $1$, which is integrable), 
$$\bbP(|C_n - D_n - \bbE[C_n - D_n \st X]| > \epsilon) \overset{p}{\to} 0,$$
and so 
\begin{align*} C_n - D_n &= \frac{1}{\sqrt{n}} \sum_{i = 1}^{n} (\mu_1(X_i) - \mu_1) - (\mu_0(X_i) - \mu_0) + o_p(1) \\ 
&\overset{p}{\to} N(0, \bbE[((\mu_1(X_i) - \mu_1) - (\mu_0(X_i) - \mu_0))^2]). 
\end{align*}
An analogous subsequencing argument to the one used before implies that we must have
\begin{align*} 
\sqrt{n}(\hat{\tau}_n - \tau) &= A_n - B_n + C_n - D_n  
\tod N(0, \nu^2),
\end{align*}
where 
\begin{align*}
\nu^2 &= 2\bbE[\Var(Y_i(1) \st X_i) + \Var(Y_i(0) \st X_i)] + \bbE[((\bbE[Y_i(1) \st X_i] - \bbE[Y_i(1)]) - (\bbE[Y_i(0) \st X_i] - \bbE[Y_i(0)]))^2] \\ 
&= 2\bbE[\sigma_1^2(X_i)] + 2\bbE[\sigma_0^2(X_i)] + \bbE[((\mu_1(X_i) - \mu_1) - (\mu_0(X_i) - \mu_0))^2].
\end{align*} 
By the law of total variance, $\bbE[\sigma^2_1(X_i)] = \sigma_1^2 - \Var(\mu_1(X_i))$ and analogously for $\bbE[\sigma^2_0(X_i)]$. Moreover, the last expectation in the formula for $\nu^2$ above is just $\Var(\mu_1(X_i) - \mu_0(X_i))$. Then 
\begin{align*}
\nu^2 &= 2\sigma_1^2 + 2\sigma_0^2 - 2 \Var(\mu_1(X_i)) - 2\Var(\mu_0(X_i)) + \Var(\mu_1(X_i) - \mu_0(X_i)) \\ 
&= 2\sigma_1^2 + 2\sigma_0^2 - \Var(\mu_1(X_i)) - \Var(\mu_0(X_i)) - 2 \Cov(\mu_1(X_i), \mu_0(X_i)) \\ 
&= 2\sigma_1^2 + 2\sigma_0^2 - \Var(\mu_1(X_i) + \mu_0(X_i)) \\ 
&= 2\sigma_1^2 + 2\sigma_0^2 - \Var(g(X_i)) \\ 
&= 2\sigma_1^2 + 2\sigma_0^2 - \bbE[(\bbE[g(X_i)] - g(X_i))^2],
\end{align*}
completing the proof.
\end{proof}

\begin{proof}[\bf Proof of Theorem \ref{thm:consistency}]
To prove consistency of $\hat{\nu}^2 = \hat{a}_n^2 - (\hat{b}_n^2 + \hat{\tau}_n^2)/2$ for $\nu^2$, we prove consistency of each term separately for a certain component of $\nu^2$. We break this up into three lemmas that correspond to Lemmas S.1.5, S.1.6, and S.1.7 of \cite{bai-romano2022}. The proofs are nearly identical but just require us to rederive a few initial results using our slightly different assumptions. 

We begin by stating a useful lemma that the first two sample moments converge to their population analogs. 

\begin{lemma} \label{lemma:moment-consistency}
For $z \in \{0, 1\}$, $\hat{\mu}_n(z) \overset{p}{\to} \bbE[Y_i(z)]$ and $\hat{\sigma}_n^2(z) \overset{p}{\to} \Var(Y_i(z))$, where 
\begin{align*}
    \hat{\mu}_n(z) &= \frac{2}{n} \sum_{i: Z_i = z} Y_i \quad\text{and}\quad
    \hat{\sigma}_n^2(z) = \frac{2}{n} \sum_{i: Z_i = z} (Y_i - \hat{\mu}_n(z))^2. 
\end{align*}
\end{lemma}

\begin{proof}
It suffices to show that 
$$\frac{2}{n} \sum_{i = 1}^{n} Y_i^r \indic\{Z_i = z\} \overset{p}{\to} \bbE[Y_i^r(z)]$$
for $r \in \{1, 2\}$ and $z \in \{0, 1\}$. We do so for the case $r = 1$ and $d = 1$ since the argument in each case is analogous. Note that 
\begin{align*} 
\hat{\mu}_n(1) &= \frac{2}{n} \sum_{i = 1}^{n} Y_i(1) Z_i \\ 
&= \frac{2}{n} \sum_{i = 1}^{n} Z_i(Y_i(1) - \mu_1(X_i)) + \frac{2}{n} \sum_{i = 1}^{n} Z_i \mu_1(X_i) \\ 
&= \frac{2}{n} \sum_{i = 1}^{n} Z_i(Y_i(1) - \mu_1(X_i)) \\ 
&\quad + \frac{1}{n} \sum_{i = 1}^{n} \mu_1(X_i) + \frac{1}{n} \left( \sum_{i : Z_i = 1} \mu_1(X_i) - \sum_{i : Z_i = 0} \mu_1(X_i) \right). 
\end{align*}
By Assumption \ref{asymptotic-assumptions2}, the last term in the final line above is $o_p(1)$. In addition, the second-to-last term converges in probability to $\mu_1 = \bbE[Y_i(1)]$ by the law of large numbers. All that remains is to show that 
$$\frac{2}{n} \sum_{i = 1}^{n} Z_i(Y_i(1) - \mu_1(X_i)) \overset{p}{\to} 0.$$
This step is identical to the proof of Lemma S.1.5 in \cite{bai-romano2022}.
\end{proof}

\begin{lemma} \label{lemma:a_n-consistency}
With $\hat{a}_n^2$ as in Theorem \ref{thm:consistency},
$$\hat{a}_n^2 \overset{p}{\to} \bbE[\sigma_1^2(X_i)] + \bbE[\sigma_2^2(X_i)] + \bbE[(\mu_1(X_i) - \mu_0(X_i))^2].$$ 
\end{lemma}

\begin{proof} 
Note that 
$$\hat{a}_n^2 = \frac{2}{n} \sum_{i = 1}^{n/2} (Y_{\pi(i)} - Y_{\pi(i + n/2)})^2 = \frac{2}{n} \sum_{i = 1}^{n} Y_i^2 - \frac{4}{n} \sum_{i = 1}^{n/2} Y_{\pi(i)} Y_{\pi(i + n/2)}.$$
By Lemma \ref{lemma:moment-consistency}, the first term converges in probability to $\bbE[Y_i(1)^2] + \bbE[Y_i(0)^2]$. We claim that the second term converges in probability to $-2\bbE[\mu_1(X_i)\mu_0(X_i)]$. This would imply the stated result, since then
\begin{align*} 
\hat{a}_n^2 &\overset{p}{\to} \bbE[Y_i(1)^2] + \bbE[Y_i(0)^2] - 2\bbE[\mu_1(X_i)\mu_0(X_i)] \\ 
&= \left(\bbE[Y_i(1)^2] - \bbE[\mu_1(X_i)^2]\right) + \left(\bbE[Y_i(0)^2] - \bbE[\mu_0(X_i)^2]\right)  + \bbE[\mu_1(X_i)^2 + \mu_0(X_i)^2 - 2\mu_0(X_i)\mu_1(X_i)] \\ 
&= \bbE[\Var(Y_i(1) \st X_i)] + \bbE[\Var(Y_i(0) \st X_i)] + \bbE[(\mu_1(X_i) - \mu_0(X_i))^2]
\end{align*}
by the law of total variance. Note that 
\begin{align*} 
\bbE\Biggl[\frac{4}{n} \sum_{i = 1}^{n/2} Y_{\pi(i)}Y_{\pi(i + n/2)} \Bigm| X \Biggr] &= \frac{2}{n} \sum_{i = 1}^{n/2} \mu_1(X_{\pi(i)})\mu_0(X_{\pi(i + n/2)}) + \mu_0(X_{\pi(i)})\mu_1(X_{\pi(i + n/2)}) \\ 
&= \frac{2}{n} \sum_{i = 1}^{n} \mu_1(X_{\pi(i)})\mu_0(X_{\pi(i)}) \\ 
&\quad + \frac{2}{n} \sum_{i = 1}^{n/2} (\mu_1(X_{\pi(i)}) - \mu_1(X_{\pi(i + n/2)}))(\mu_0(X_{\pi(i)}) - \mu_0(X_{\pi(i + n/2)})). 
\end{align*}
By Assumption \ref{pairing-assumption}, the second sum is $o_p(1)$, and so 
$$\bbE\Biggl[\frac{4}{n} \sum_{i = 1}^{n/2} Y_{\pi(i)}Y_{\pi(i + n/2)} \Bigm| X \Biggr] \overset{p}{\to} 2\bbE[\mu_1(X_i)\mu_0(X_i)].$$
All that remains is to show that 
$$\frac{4}{n} \sum_{i = 1}^{n/2} Y_{\pi(i)}Y_{\pi(i + n/2)} - \bbE[Y_{\pi(i)}Y_{\pi(i + n/2)} \st X] \overset{p}{\to} 0.$$
This step is identical to the second half of Lemma S.1.6 in \cite{bai-romano2022}.
\end{proof}

\begin{lemma} \label{lemma:b_n-consistency}
With $\hat{b}_n^2$ as in Theorem \ref{thm:consistency}, 
$$\hat{b}_n^2 \overset{p}{\to} \bbE[(\mu_1(X_i) - \mu_0(X_i))^2].$$
\end{lemma} 

\begin{proof} 
Recall that 
$$\hat b_n^2=\frac{4}{n} \sum_{i = 1}^{n/2} Y_{\pi(2i)} Y_{\pi(2i - 1)} - \frac{4}{n} \sum_{i = 1}^{n/2} Y_{\pi(i)} Y_{\pi(i + n/2)}.$$
We have seen in Lemma \ref{lemma:a_n-consistency} that the second term converges in probability to $-2\bbE[\mu_1(X_i)\mu_0(X_i)]$. To complete the proof, it suffices to show that the first term converges in probability to $\bbE[\mu_1(X_i)^2] + \bbE[\mu_0(X_i)^2]$. Each summand contains two units in the same half of the knapsack partition, so they have the same treatment status. Then for $i \leq n/2$, 
$$\bbE[Y_{\pi(2i)} Y_{\pi(2i - 1)} \st X] = \frac{1}{2} \mu_1(X_{\pi(2i)}) \mu_1(X_{\pi(2i - 1)}) + \frac{1}{2} \mu_0(X_{\pi(2i)}) \mu_0(X_{\pi(2i - 1)})$$
and so 
$$\bbE\Biggl[\frac{4}{n} \sum_{i = 1}^{n/2} Y_{\pi(2i)} Y_{\pi(2i - 1)} \bigm| X\Biggr] = \frac{2}{n} \sum_{i = 1}^{n/2} \mu_1(X_{\pi(2i)}) \mu_1(X_{\pi(2i - 1)}) + \frac{2}{n} \sum_{i = 1}^{n/2} \mu_0(X_{\pi(2i)}) \mu_0(X_{\pi(2i - 1)}).$$
We claim that the first sum converges in probability to $\bbE[\mu_1(X_i)^2]$ and the second converges in probability to $\bbE[\mu_0(X_i)^2]$. For the first, write 
\begin{align*} \frac{2}{n} \sum_{i = 1}^{n/2} \mu_1(X_{\pi(2i)}) \mu_1(X_{\pi(2i - 1)}) &= \frac{2}{n} \sum_{i = 1}^{n/2} \mu_1(X_{\pi(2i)})^2 + \frac{2}{n} \sum_{i = 1}^{n/2} \mu_1(X_{\pi(2i)})\bigl(\mu_1(X_{\pi(2i - 1)}) - \mu_1(X_{\pi(2i)})\bigr) \\ 
&= \frac{1}{n} \sum_{i = 1}^{n} \mu_1(X_{\pi(i)})^2 + \frac{1}{n} \sum_{i = 1}^{n/2} \mu_1(X_{\pi(2i)})^2 - \mu_1(X_{\pi(2i - 1)})^2 \\ 
&\quad + \frac{2}{n} \sum_{i = 1}^{n/2} \mu_1(X_{\pi(2i)})\bigl(\mu_1(X_{\pi(2i - 1)}) - \mu_1(X_{\pi(2i)})\bigr) \\ 
&=\frac{1}{n} \sum_{i = 1}^{n} \mu_1(X_{\pi(i)})^2 - \frac{1}{n} \sum_{i = 1}^{n/2} \bigl(\mu_1(X_{\pi(2i)}) - \mu_1(X_{\pi(2i -1 )}\bigr)^2.
\end{align*} 
The first term in the final expression converges in probability to $\bbE[\mu_1(X_i)^2]$ by the law of large numbers. The second converges in probability to zero by Assumption \ref{pairs-of-pairs-assumption}. To conclude, one need only show that $\hat{b}_n^2 - \bbE[\hat{b}_n^2 \st X] \overset{p}{\to} 0$. This is the same as in Lemma S.1.7 of \cite{bai-romano2022}. 
\end{proof}
We can now complete the proof of Theorem \ref{thm:consistency}. Recall that $\hat{\nu}^2 = 2\hat{a}_n^2 - (\hat{b}_n^2 + \hat{\tau}_n^2)$ and
\begin{align*} 
\hat{a}_n^2 &\overset{p}{\to} \bbE[\sigma_1^2(X_i)] + \bbE[\sigma_0^2(X_i)] + \bbE[(\mu_1(X_i) - \mu_0(X_i))^2], \\ 
\hat{b}_n^2 &\overset{p}{\to} \bbE[(\mu_1(X_i) - \mu_0(X_i))^2], \\ 
\hat{\tau}_n^2 &\overset{p}{\to} (\mu_1 - \mu_0)^2.
\end{align*} 
By the continuous mapping theorem, 
\begin{align*} 
\hat{\nu}^2 &= 2\hat{a}_n^2 - (\hat{b}_n^2 + \hat{\tau}_n^2) \\ 
&\overset{p}{\to} 2\bbE[\sigma_1^2(X_i)] + 2\bbE[\sigma_0^2(X_i)] + 2\bbE[(\mu_1(X_i) - \mu_0(X_i))^2] - \left(\bbE[(\mu_1(X_i) - \mu_0(X_i))^2] + (\mu_1 - \mu_0)^2 \right) \\ 
&= 2\bbE[\sigma_1^2(X_i)] + 2\bbE[\sigma_0^2(X_i)] + \bbE[(\mu_1(X_i) - \mu_0(X_i))^2] - (\mu_1 - \mu_0)^2 \\ 
&= 2\bbE[\sigma_1^2(X_i)] + 2\bbE[\sigma_0^2(X_i)] + \Var(\mu_1(X_i) - \mu_0(X_i)) \\ 
&= 2\bbE[\sigma_1^2(X_i)] + 2\bbE[\sigma_0^2(X_i)] + \bbE[(\tau(X_i) - \tau)^2] \\ 
&= \nu^2.
\end{align*} 
Hence, $\hat{\nu}^2 \overset{p}{\to} \nu^2$, and so 
$$\frac{\sqrt{n}(\hat{\tau}_n - \tau)}{\hat{\nu}} \tod N(0, 1),$$
by Slutsky's theorem and Theorem \ref{thm:CLT}.
\end{proof}

\begin{lemma} \label{lemma:max-o_p}
Suppose that $Y_1, Y_2, \ldots$ are sampled IID from a distribution $P$ with finite variance. Then 
$$\underset{1 \leq i \leq n}{\max} \text{ } Y_i = o_p(n^{1/2}).$$
\end{lemma}
\begin{proof} 
We will prove the stronger result that
$$\limn \bbE\left[\underset{1 \leq i \leq n}{\max} \text{ } \frac{Y_i^2}{n} \right] = 0.$$
First, observe that
\begin{align*} 
\bbE\left[\underset{1 \leq i \leq n}{\max} \text{ } \frac{Y_i^2}{n} \right] &= \int_0^{\infty} n^{-1} \bbP(\underset{1 \leq i \leq n}{\max} \text{ } Y_i^2 > t) dt \\ 
&:= \int_0^{\infty} f_n(t) dt. 
\end{align*} 
Here, $f_n(t)$ converges pointwise to $0$ everywhere as $n \to \infty$, since it is bounded by $1/n$. To apply the dominated convergence theorem, we need to find an integrable function that dominates it. Indeed, 
\begin{align*} 
f_n(t) &= \frac{1}{n} \bbP(\underset{1 \leq i \leq n}{\max} \text{ } Y_i^2 > t) \\ 
&\leq \frac{1}{n} \sum_{i = 1}^{n} \bbP(Y_i^2 > t) \\ 
&= \bbP(Y_1^2 > t)
\end{align*} 
Finally, $\bbP(Y_1^2 > t)$ is integrable by the assumption that $\Var(Y_i) < \infty$. The dominated convergence theorem then gives the desired result. 
\end{proof}

\section{Additional simulations}\label{sec:add_sims}
In this section, we present more simulation results in a variety of settings to highlight the strengths and weaknesses of the correlated Bernoulli designs. We first present more examples like that shown in Figure \ref{fig:noiseless} in which $g(X)$ is known. Then, we discuss an example in which only $X$ is known and there is no prior information about $g(X)$. The correlated Bernoulli design is more poorly-suited to this case, but we present a reasonable approach and assess its performance. 

We focus on a few different regimes in which we can expect correlated Bernoulli designs to perform especially well or poorly compared to the optimal stratified design. This involves the size of $g^{\top} \Sigma_{\bern} g$ relative to $g^{\top} \Sigma_{\strat} g$. In other words, we should expect a strong improvement over the optimal stratified design when the knapsack difference is small and the $g_i$ terms have a few large gaps, and a weaker improvement when the knapsack difference is larger or the $g_i$ terms are relatively equispaced. Of note, when $g(X)$ has a symmetric distribution, it is likely that the knapsack difference will be small since all pairs of symmetrically opposed points sum to the same value. 

We thus consider the following settings: \\ 
\begin{compactenum}[\qquad1)]
\item $X \sim \text{Unif}[0, 1]$, $\mu_0(X) = X$, $\mu_1(X) = \mu_0(X) + 0.05 X^{1.05}$,\\[-2.5ex]
\item $X \sim N(0, 1)$, $\mu_0(X) = X^3$, $\mu_1(X) = 1.2 \mu_0(X)$, and \\[-2.5ex]
\item $X \sim \text{Poisson}(20)$, $\mu_0(X) = X^2$, $\mu_1(X) = \mu_0(X) + 0.2 X^3$.
\end{compactenum} 
\text{ } \\ 
In each case, $Y \st X, Z$ is generated IID as $N(\mu_Z(X), 1)$, so the irreducible error (i.e., the design-independent term in the variance decomposition in Lemma \ref{lemma:equiv-objective}) is $4/n$. The first setting should have a small knapsack difference since $g(X)$ is nearly symmetric and should have a small $g^{\top} \Sigma_{\strat} g$ since $g(X)$ has fairly equispaced gaps. The second should have a small knapsack difference because of symmetry but a larger $g^{\top} \Sigma_{\strat} g$ since the gaps $g_{(2i)} - g_{(2i - 1)}$ are quite large. The third should have higher values for both because of the large right tail in the distribution of $g(X)$. 

The results of these simulations are shown in Figure \ref{fig:additional_sims}. As predicted, all designs achieve nearly the fundamental lower bound of $4/n$ in the first setting, and so the correlated Bernoulli design provides negligible improvement. In the second setting, the correlated Bernoulli design is near the fundamental lower bound and provides a large improvement over the stratified design, which suffers from wide gaps in the $g$ vector. In the Poisson setting, both approaches have a performance gap to the fundamental lower bound, reflecting a larger knapsack difference. Still, the performance improvement of the correlated Bernoulli designs over the stratified ones is significant. 

We next assess the performance of the designs in these simulation settings when fit using the proxy $h = \mu_0$. This is a reasonable scenario to arise in practice, when there may be a large dataset of prior control units before a treatment is rolled out. If $\tau(X)$ is small, then $g(X) \approx 2\mu_0(X)$ and so $\mu_0(X)$ is a good proxy. 

These results are shown in Figure \ref{fig:additional_sims_proxy}. Unsurprisingly, the performances remain similar in the uniform and Gaussian simulations, in which $\tau(X)$ is fairly small. However, in the Poisson setting in which $\mu_0$ is a poor proxy, the prior performance ranking fails to hold; in fact, the optimal correlated Bernoulli design with a single coin flip does quite poorly for large $n$. Interestingly, the hybrid SIB design retains strong performance; this suggests some degree of stability to a poorly chosen proxy, as the theory predicts. 

\clearpage
\begin{figure}[hbt!]
    \centering
    \begin{subfigure}{\textwidth}
        \centering
        \includegraphics[width=0.6\textwidth]{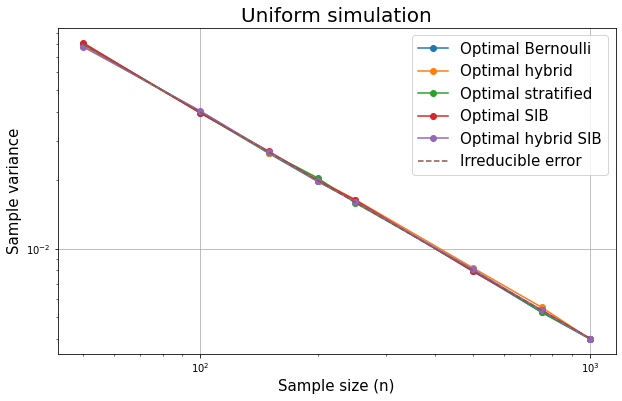}
        \label{fig:sub1}
    \end{subfigure}
    
    \vspace{1em}
    
    \begin{subfigure}{\textwidth}
        \centering
        \includegraphics[width=0.6\textwidth]{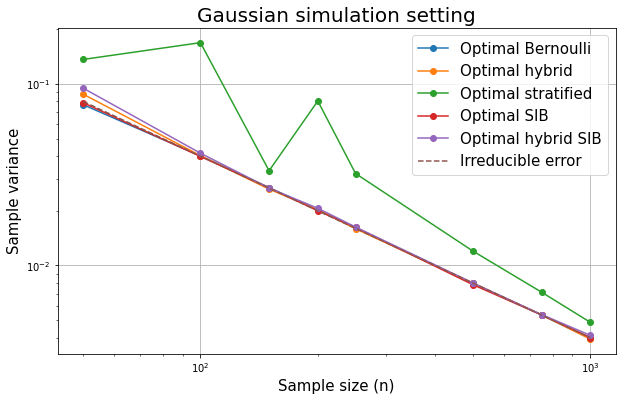}
        \label{fig:sub2}
    \end{subfigure}
    
    \vspace{1em} 
    
    \begin{subfigure}{\textwidth}
        \centering
        \includegraphics[width=0.6\textwidth]{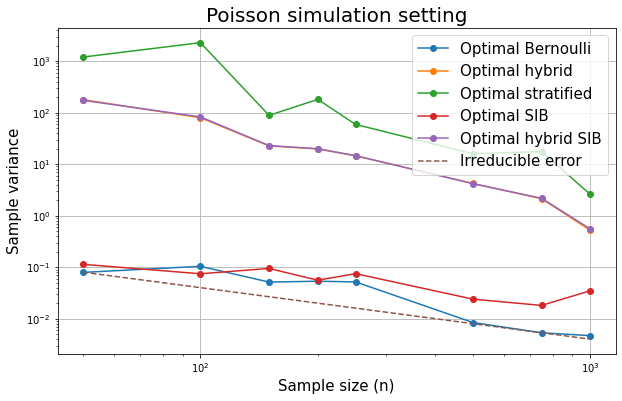}
        \label{fig:sub3}
    \end{subfigure}
    
\caption{Comparison of designs for each of the three additional simulation settings in Appendix \ref{sec:add_sims}.}
\label{fig:additional_sims}
\end{figure}

\clearpage
\begin{figure}[hbt!]
    \centering
    \begin{subfigure}{\textwidth}
        \centering
        \includegraphics[width=0.6\textwidth]{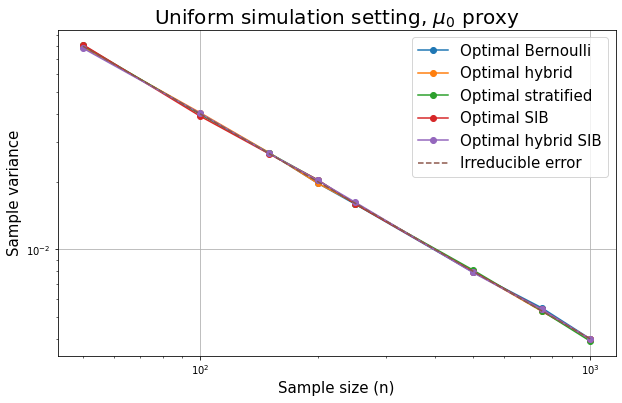}
        \label{fig:sub1_proxy}
    \end{subfigure}
    
    \vspace{1em}
    
    \begin{subfigure}{\textwidth}
        \centering
        \includegraphics[width=0.6\textwidth]{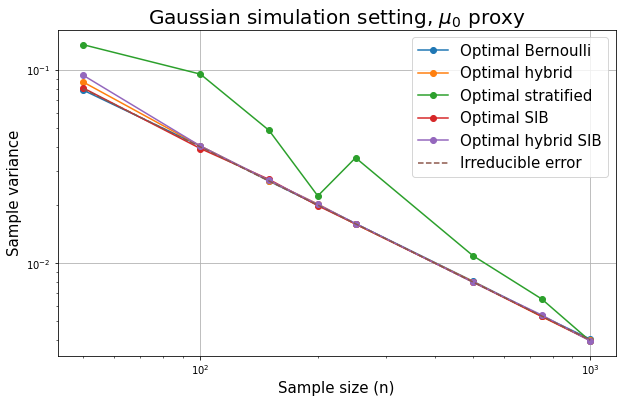}
        \label{fig:sub2_proxy}
    \end{subfigure}
    
    \vspace{1em} 
    
    \begin{subfigure}{\textwidth}
        \centering
        \includegraphics[width=0.6\textwidth]{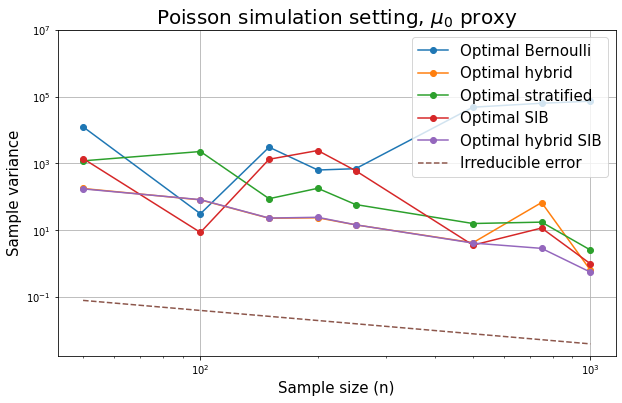}
        \label{fig:sub3_proxy}
    \end{subfigure}
    
\caption{Comparison of designs for each of the three additional simulation settings in Appendix \ref{sec:add_sims} when fit using the proxy $h = \mu_0$.}
\label{fig:additional_sims_proxy}
\end{figure}

\end{document}